\documentclass[a4paper,10pt,  pointlessnumbers]{article}
\usepackage{geometry}
\usepackage[T1]{fontenc}        

\usepackage[utf8x]{inputenc}    
\usepackage{ucs}                
\usepackage{amsmath}            
\usepackage{amsfonts}
\usepackage{amssymb}
\usepackage{amsthm} 
\usepackage{thmtools}
\usepackage[authoryear]{natbib}
\usepackage{ifthen}
\usepackage{color}
\usepackage{xcolor}	
\usepackage{colortbl}
\usepackage{graphicx} 
\usepackage{textcomp}            
\usepackage{palatino}            
\usepackage{hhline}

\usepackage{varioref}            
\usepackage{url}                
\usepackage{cleveref}             
\usepackage{fancyhdr}
\usepackage{float}
\usepackage{multirow}
\usepackage{rotating}
\usepackage{subfig}
\usepackage{titlesec}
\usepackage{setspace}
\usepackage{tabularx}
\usepackage{booktabs}
\definecolor{tugreen}{HTML}{84b819}
\definecolor{pgrey}{rgb}{242, 242, 242}
\definecolor{ashgrey}{rgb}{0.7, 0.75, 0.71}
 	\definecolor{gainsboro}{rgb}{0.86, 0.86, 0.86}
 	\definecolor{lightyellow}{rgb}{1.0, 1.0, 0.88}
 	\definecolor{pastelgray}{rgb}{0.81, 0.81, 0.77}
 	\definecolor{timberwolf}{rgb}{0.86, 0.84, 0.82}
 	\definecolor{beige}{rgb}{0.96, 0.96, 0.86}
\titleformat{\chapter}{\fontfamily{qzc}\selectfont \Huge \color{tugreen}}{\thechapter}{1ex}{}
\titlespacing{\chapter}{0pt}{-2ex}{*1.5}
\usepackage{xr} 
\usepackage{pdfpages}
\usepackage{xcite}

\newcommand{\R}{{\mathbb{R}}}         
\newcommand{\E}{{\mathbb{E}}}

\newcommand{\bay}{\begin{array}}
\newcommand{\eay}{\end{array}}

\newcommand{\bqa}{\begin{eqnarray*}}
\newcommand{\eqa}{\end{eqnarray*}}

\newcommand{\bee}{\begin{eqnarray*}}
\newcommand{\eee}{\end{eqnarray*}}

\newcommand{\bea}{\begin{eqnarray*}}
\newcommand{\eea}{\end{eqnarray*}}

\newcommand{\bqan}{\begin{eqnarray}}
\newcommand{\eqan}{\end{eqnarray}}

\newcommand{\be}{\begin{eqnarray}}
\newcommand{\ee}{\end{eqnarray}}

\newcommand{\bit}{\begin{itemize}}
\newcommand{\eit}{\end{itemize}}

\newcommand{\ben}{\begin{enumerate}}
\newcommand{\een}{\end{enumerate}}

\newcommand{\beq}{\begin{equation}}
\newcommand{\eeq}{\end{equation}}

\newcommand{\bdes}{\begin{description}}
\newcommand{\edes}{\end{description}}

\newcommand{\btb}{\begin{tabular}}
\newcommand{\etb}{\end{tabular}}

\newcommand{\bcen}{\begin{center}}
\newcommand{\ecen}{\end{center}}

\newcommand{\bmp}{\begin{minipage}}
\newcommand{\emp}{\end{minipage}}

\newcommand{\vepsilon}{\vep}

\newcommand{\Cov}{\operatorname{{\it Cov}}}

\newcommand{\Var}{\operatorname{{\it Var}}}

\newcommand{\tr}{\operatorname{tr}}

\newcommand{\diag}{\operatorname{\it diag}}

\newcommand{\rank}{\operatorname{\it rank}}

\newcommand{\vc}{\boldsymbol{c}}

\newcommand{\vf}{\boldsymbol{f}}

\newcommand{\vx}{\boldsymbol{x}}
\newcommand{\vy}{\boldsymbol{y}}

\newcommand{\vA}{\boldsymbol{A}}

\newcommand{\vC}{\boldsymbol{C}}

\newcommand{\vI}{\boldsymbol{I}}

\newcommand{\vM}{\boldsymbol{M}}

\newcommand{\vP}{\boldsymbol{P}}
\newcommand{\vQ}{\boldsymbol{Q}}
\newcommand{\vR}{\boldsymbol{R}}

\newcommand{\vT}{\boldsymbol{T}}

\newcommand{\vX}{\boldsymbol{X}}
\newcommand{\vY}{\boldsymbol{Y}}
\newcommand{\vZ}{\boldsymbol{Z}}

\newcommand{\vbeta}{\boldsymbol{\beta}}

\newcommand{\vep}{\boldsymbol{\epsilon}}

\newcommand{\vmu}{\boldsymbol{\mu}}

\newcommand{\vSigma}{\boldsymbol{\Sigma}}

\newcommand{\veins}{{\bf 1}}
\newcommand{\vnull}{{\bf 0}}

\newcommand{\vUpsilon}{\boldsymbol{\Upsilon}}

\newcommand{\ind}{1\hspace{-0.7ex}1}

\newcommand{\Hypo}{\mathcal{H}_0}

\newtheoremstyle{Test1}
  {2 \baselineskip}
  {1.5 \baselineskip}
  {\itshape}
  {-0.0ex}
  {\fontfamily{ppl}\fontseries{l}\fontshape{n}}
  {:}
  {\newline}
   {}

\theoremstyle{Test1}
\newtheorem{Sa}{Theorem}[section]
\newtheorem{theorem}{Theorem}[section]
\newtheorem{re}[Sa]{Remark}
 
\newtheorem{Le}[Sa]{Lemma}

\newcommand{\Lan}{\mathcal{O}}

\newcolumntype{x}[1]{!{\centering\arraybackslash\vrule width #1}}

\makeatletter
\renewenvironment{proof}[1][\proofname]{\par
  \pushQED{\qed}%
  \fontfamily{ppl}\fontseries{m}\fontshape{it} \topsep6\p@\@plus6\p@\relax
  \trivlist
  \item[\hskip\labelsep
        \bfseries
    #1\@addpunct{:}]\ignorespaces
}{%
  \popQED\endtrivlist\@endpefalse
}
\makeatother

\usepackage[skip=2pt plus1pt, indent=10pt]{parskip}
\begin{document}


\title{\Large \bf Quadratic Form based Multiple Contrast Tests for Comparison of Group Means}
\author{Paavo Sattler$^{1}$, Markus Pauly$^{1,2}$ and Merle Munko$^{3}$
\\[0.4ex] }

\vspace{-8ex}
  \date{}
\maketitle \vspace*{0.25ex}
\begin{center}\noindent${}^{1}$ {Department of Statistics, TU Dortmund University, Joseph-von-Fraunhofer-Straße 2-4, 44227 Dortmund, Germany\\\mbox{ }\hspace{1 ex}email: paavo.sattler@tu-dortmund.de }\\
 \noindent${}^{2}$ {{Research Center Trustworthy Data Science and Security, UA Ruhr, Joseph-von-Fraunhofer-Straße 25, 44227 Dortmund, Germany}}\\
  \noindent${}^{3}$ {Department of Mathematics, Otto-von-Guericke University Magdeburg, Germany}
\vspace*{8.25ex}
 \end{center}
\begin{abstract}
Comparing the mean vectors across different groups is a cornerstone in the realm of multivariate statistics, with quadratic forms commonly serving as test statistics. 
However, when the overall hypothesis is rejected, identifying specific vector components or determining the groups among which differences exist requires additional investigations.

Conversely, employing multiple contrast tests (MCT) allows conclusions about which components or groups contribute to these differences. However, they come with a trade-off, as MCT lose some benefits inherent to quadratic forms.
In this paper, we combine both approaches to get a quadratic form based multiple contrast test that leverages the advantages of both. To understand its theoretical properties, we investigate its asymptotic distribution in a semiparametric model. 
We thereby focus on two common quadratic forms — the Wald-type statistic and the ANOVA-type statistic — although our findings are applicable to any quadratic form. 

Furthermore, we employ Monte-Carlo and resampling techniques to enhance the test's performance in small sample scenarios.
Through an extensive simulation study, we assess
the finite sample properties of our proposed tests against existing alternatives, highlighting their advantages.
\end{abstract}


\vfill
\vfill

\section{Motivation and introduction}\label{int}

In multivariate analysis, comparing the mean vectors of two or more groups is one of the most common 
null hypothesis, also known as one-way MANOVA. 
As traditional methods like Hotelling's $T^2$ or Wilk's Lambda tests \citep{johnson2007} rely on restrictive assumptions like normality and covariance homogeneity, we have seen the development of alternative approaches to address these limitations. For example, \cite{krishnamoorthy2010} or \cite{smaga2017} proposed methods for the two sample case while \cite{konietschke2015,roy2015}, \cite{smaga2017}, \cite{hu2017,friedrich2017a} or \cite{sattler2018} studied one-way and more
general hypotheses in complex MANOVA designs without postulating normality or equal covariance matrices. 

An appropriate way to investigate hypotheses regarding vector-valued parameters, like expectation vectors, are quadratic forms. This already holds for the common Hotelling's $T^2$ and Wilk's Lambda tests but also applies for most of the generalizations mentioned above as well as other multivariate hypotheses, e.g. in completely nonparametric settings \citep{brunner2016,Brunner2019,dobler2020,rubarth2022}. 
In these papers, only one or a few specific quadratic forms are examined. More recently, however, \cite{sattler2021} and \cite{baumeister2024} have expanded the analysis to encompass a broader range of general quadratic forms. We will also do this in the present paper.

One of the key strengths of quadratic forms as test statistics is that they yield a univariate value out of multivariate input. However, this also introduces a significant limitation: When the hypothesis of equal mean vectors is rejected,  the univariate test statistic does not reveal which specific components are responsible for the differences. As a result, this often necessitates further comparisons with multiplicity adjustments, a process that can be notably inefficient.

In contrast, so-called multiple contrast tests (MCT) offer a compelling approach by not only detecting differences in parameters but, at the same time, draw conclusions where the differences occur. MCT usually use a maximum of univariate t-type statistics as test statistics and equicoordinate quantiles from multivariate normal, t- or resampling distributions as critical values, see, e.g., \cite{hasler2008,konietschke2012}, \cite{konietschke2013,gunawardana2019,umlauft2019,noguchi2020} for some examples. However, for multivariate settings, the component-wise dependence is only taken into account through the estimated correlation of these t-test statistics, and not within the test statistics themselves.
Thus, it would be beneficial to integrate quadratic forms with MCT as this combination could yield a test procedure that leverages the advantages of both. 

The present paper aims to propose such a quadratic form based MCT. Thereby, the procedures will allow for general quadratic forms and specifically address the question of where differences occur if the global hypothesis is rejected. For ease of presentation, we will initially focus on the one-way hypothesis of equal group expectation vectors with subsequent identification of responsible components or groups. However, we later explain how the method can be extended to more general settings and hypotheses. 
We examine the large sample properties of all procedures by means of asymptotic theory. Here, it turns out that the commonly used equicoordinate quantiles for MCT with t-type statistics will not lead to appropriate critical values for MCT based on quadratic forms, making adoptions of other approaches necessary \citep{beran1988balanced, buehlmann1998,munko2024,munko2023}. The small sample performance of the resulting tests will be investigated in Monte Carlo simulations.

The paper is organized as follows: 
In the subsequent section, the statistical model will be introduced, together with the underlying assumptions. Afterwards, focusing on the equality of expectation vectors, the basic concept of multiple contrast tests and the test statistics are proposed, and the asymptotic distributions are derived (Section~\ref{Equality}). Subsequently, this is generalized for general partitions of global hypotheses without concentrating on the responsible components for rejection. In Section~\ref{Resampling}, a resampling strategy is used to generate quantiles with the aim to improve our tests' small sample behaviour. This is demonstrated in Section~\ref{Simulations}, which contains simulation results on type-I-error control and power for different alternatives, followed by an illustrative application for a real EEG-data set (Section~\ref{Example}). In \Cref{Conclussion} we finally draw conclusions and mention potential future applications. All proofs are deferred to a technical supplement, where additional details and simulations can be found.

\section{Model}\label{Model}
To ensure broad applicability, we consider a general semiparametric model given  by independent $d$-dimensional random vectors 
\[\vX_{ik}= \vmu_i + \vep_{ik}.\]
Here, the index $k=1, \dots,n_i$ represents the subject on which $d$-variate observations are measured while the index $i=1, \dots, a$ refers to the respective group. Moreover, $\vmu_i = (\mu_{ij})_{j=1}^d$ denotes the expectation vector in group $i$ while $\vep_{ik} = (\epsilon_{ikj})_{j=1}^d$ is the corresponding random error of subject $k$ in that group. 
We study situations with $a\geq 2$ groups and assume that for fixed $i\in\{1,...,a\}$ and $j\in\{1,...,d\}$ the random variables $\epsilon_{ikj}$ are identical distributed with $\E(\epsilon_{ikj})=0$. It is important to note that by splitting up the indices, this setting implies factorial designs inside a group as in \cite{konietschke2015}.

We allow unbalanced sample sizes $n_1,...,n_a$ and denote with $N=\sum_{i=1}^a n_i$ the total sample size. For our asymptotic investigations, we employ two additional assumptions:\vspace{0.1cm}
\begin{itemize}
    \item[(A1)] $\frac{n_i}{N}\to \kappa_i\in(0,1],\quad i=1,...,a$\vspace{0.1cm}
    \item[(A2)] $\vSigma_i=\Cov(\vepsilon_{ik})\geq 0$,\quad $i=1,...,a$.\vspace{0.1cm}
\end{itemize}
Here, (A1) 
ensures that no group is negligible, while assumption (A2) guarantees the existence of second moments. 
Within our test statistics, we use different kinds of pooled mean vectors, a group-wise mean vector $\overline \vX_{i\cdot}=\frac{1}{n_i}\sum_{k=1}^{n_i}\vX_{ik}$ with corresponding $\overline \vX=(\overline \vX_{1\cdot}^\top,...,\overline \vX_{a\cdot}^\top)^\top$ and a component-wise mean vector 
\[\overline \vX^{(j)}=\left(\frac{1}{n_1}\sum\limits_{k=1}^{n_1}X_{1kj},...,\frac{1}{n_a}\sum\limits_{k=1}^{n_a}X_{akj}\right)^\top.\]
{Here and in the following the superscript ${}^\top$ indicates a transposed matrix or vector.}

Assumptions (A1) and (A2) guarantee that the normalized version of the latter $\sqrt{N}\ \overline \vX^{(j)}$ possess the asymptotic block-diagonal covariance matrix $\vSigma^{(j)}:=\diag(\kappa_1^{-1}(\vSigma_{1})_{jj},...,\kappa_a^{-1}(\vSigma_{a})_{jj})=\lim\limits_{N\to \infty}\Var(\sqrt{N} \overline \vX^{(j)})\geq 0$ for each $j=1,\dots,d$.

In the following, we use $\stackrel{\mathcal{P}}{\to}$ to denote convergence in probability and $\stackrel{\mathcal{D}}{\to}$ for convergence in distribution, both as $N\to \infty$.

\section{Global and Multiple Testing in one-way MANOVA}\label{Equality}

We first consider the one-way MANOVA hypothesis 
\[\mathcal H_0:
(\mu_{1j})_{j=1}^{d}=\dots=(\mu_{aj})_{j=1}^{d}
\]
of equal expectation vectors 
$\vmu_{i}=(\mu_{ij})_{j=1}^{d}$. 
We refer to it as \emph{global hypothesis} as it is the intersection of the local (component-wise) hypotheses $\mathcal H_0^{(j)}:\mu_{1j}=\dots=\mu_{aj}$, $j=1,\dots, d$. 
With the Kronecker-product and the centering matrix $\vP_a= \vI_a-\veins_a\veins_a^\top/a$, the global hypotheses can be formulated using the hypothesis matrix $\vC=\vP_a\otimes\vI_d$ and $\vmu=(\vmu_1^\top,...,\vmu_a^\top)^\top$ as  $\mathcal H_0:\vC\vmu=\vnull_{ad}$. {Thereby, $\vI_a$ denotes the $a\times a$ identity matrix and $\veins_a$ is an $a$-dimensional vector of $1$s. }
We will later allow for general hypothesis matrices (Section~\ref{general}).

From our model assumptions, it follows that
$\vT:=\sqrt{N}\ (\overline \vX-\vmu)\stackrel{\mathcal{D}}{\to}\vZ\sim\mathcal{N}_{ad}(\vnull_{ad},\vSigma)$ 
by the multivariate central limit theorem 
with $\vSigma:=\bigoplus_{i=1}^a \kappa_i^{-1}\vSigma_i$, where $\bigoplus$ denotes the direct sum (i.e. it defines the block diagonal matrix consisting of the entries). Therefore, under the global null hypothesis $\mathcal H_0$, it holds that $\sqrt{N}\vC\overline \vX\stackrel{\mathcal{D}}{\to}\vC\vZ\sim\mathcal{N}_{ad}(\vnull_{ad},\vC\vSigma\vC^\top)$. This motivates to consider test statistics in quadratic forms of $\sqrt{N}\vC\overline \vX$. 
Common examples cover the Wald-Type-Statistic \[WTS(\vSigma):= N\cdot(\vC\overline\vX)^\top (\vC\vSigma\vC^\top)^+(\vC\overline\vX)\] using a Moore-Penrose inverse (denoted by the superscript ${}^+$) or the ANOVA-type statistic \[ATS(\vSigma):= N\cdot (\vC\overline\vX)^\top (\vC\overline\vX)/\tr(\vC\vSigma\vC^\top)\] 
as well as their empirical counterparts, where the unknown covariance matrix is substituted by a consistent estimator $\widetilde \vSigma$, see, e.g., \cite{brunner1997,brunner2001b,konietschke2015,friedrich2017,sattler2025}.
 In our setting, the empirical covariance matrices for the i-th group $\widehat \vSigma_i= \sum_{k=1}^{n_i} (\vX_{ik} - \overline \vX_{i\cdot})(\vX_{ik} - \overline \vX_{i\cdot})^\top/(n_i-1)$, are consistent estimators for $\vSigma_i$. Thus, $\widehat \vSigma=N\cdot\bigoplus_{i=1}^a n_i^{-1} \widehat \vSigma_i$ is a suitable choice for estimating $ \vSigma$. 
 Then, we obtain asymptotic correct level $\alpha$ tests for the global hypothesis by comparing these quadratic forms with adequate critical values which correspond to the asymptotic distribution of the quadratic forms under the null hypothesis.
For example, the WTS, asymptotically follows a $\chi_{\rank(\vC)}^2$ distribution under the null \citep{akritas1997} leading to $\chi_{\rank(\vC)}^2$-quantiles while also permutation approaches showed good results \citep{ditzhaus2022}. The ATS  asymptotically follows a weighted sum of independent $\chi_1^2$ distributed random variables under the null hypothesis. Here, quantiles can, e.g.,  be computed via Monte-Carlo  \citep{sattler2023}.
For general quadratic forms, both Monte-Carlo and bootstrap techniques are in widespread use, see, e.g. \cite{sattler2022} or \cite{baumeister2024}. As we focus on general quadratic forms, we will study Monte-Carlo and bootstrap approaches in the sequel. This will lead to a variety of procedures for inferring the global hypothesis.

 However, in case of a rejection, additional investigations would be necessary to identify the responsible components. That is, we would need to test $\mathcal H_0^{(j)}$ for $j=1,\dots,d$. To avoid type-I-error cumulation, we propose multiple contrast tests in quadratic forms that automatically guarantee adequate multiplicity adjustment.

To this end, we first need test statistics for the local hypotheses $\mathcal H_0^{(j)}:\{\mu_{1j}=\dots=\mu_{aj}\} = \{
 \mu_{ij} - \overline{\mu}_{\cdot j}
 =0 \text{ for all } i=1,\dots,a 
 \}$ for $j=1,\dots,d$, where 
$\overline{\mu}_{\cdot j} = \frac{1}{a}\sum_{r=1}^a\mu_{rj}$.
Here, quadratic forms in the projected $j$-th mean vector $\mathbf{P}_a \overline \vX^{(j)}$, the natural estimator for $(\mu_{1j} - \overline{\mu}_{\cdot j},\dots,\mu_{aj} - \overline{\mu}_{\cdot j})^\top$, are reasonable test statistics. 

To not limit us to a specific class (e.g., only Wald-type statistics), we follow \cite{sattler2022} and study general quadratic forms (with $0/0:=0$)

\[\widetilde Q_{Nj}={N}\cdot \frac{[\vP_a\overline \vX^{(j)} ]^\top \widetilde \vM(\vP_a,\widetilde\vSigma^{(j)})[\vP_a\overline\vX^{(j)}]}{\sqrt{2\tr([\vP_a^\top \widetilde\vM(\vP_a,\widetilde\vSigma^{(j)})\vP_a \widetilde\vSigma^{(j)}]^2)} } \cdot \ind\left( \tr\left([\vP_a^\top \widetilde\vM(\vP_a,\widetilde\vSigma^{(j)})\vP_a \widetilde\vSigma^{(j)}]^2\right)>0\right) \]
to infer $H_0^{(j)}$.
Here, $\widetilde\vM$ 
is a function taking values in $\R^{a\times a}$ and fulfilling $\widetilde\vM^\top=\widetilde\vM$, $\widetilde\vM\geq 0$ and $\widetilde\vM(\vP_a, \widetilde\vSigma^{(j)})\stackrel{\mathcal{P}}{\to} \widetilde\vM(\vP_a,\vSigma^{(j)})$ for an estimator $\widetilde\vSigma^{(j)}\stackrel{\mathcal{P}}{\to} \vSigma^{(j)}$, where a possible choice would be $\widehat \vSigma^{(j)}=N\cdot\diag(n_1^{-1}(\widehat \vSigma_1)_{jj},...,n_a^{-1}(\widehat \vSigma_a)_{jj})$.
Setting $\widetilde\vM(\vP_a,\vSigma^{(j)}):= (\vP_a\vSigma^{(j)}\vP_a^\top)^+$ we obtain a WTS for $H_0^{(j)}$, while $\widetilde\vM(\vP_a,\vSigma^{(j)}):= \vI_a/\tr(\vP_a\vSigma^{(j)}\vP_a^\top)$ leads to the ATS. Other choices could be the modified Anova-Type-statistic (MATS) from \cite{friedrich2017a} or different ATS versions as considered in \cite{Brunner2019}.

\begin{re}
 Since $\widetilde Q_{Nj}$ is already standardized, it is unnecessary,  and also computationally inefficient, to calculate a standardized version of the ATS. We can therefore work with the quadratic form version using  $\widetilde\vM(\vP_a,\widetilde \vSigma^{(j)})=\vI_{a}$, and similar simplifications are also possible for other quadratic forms.
\end{re}

The subsequent Lemma derives the marginal and joint limit distributions of $\widetilde Q_{N1},\dots,\widetilde Q_{Nd}$.

\begin{Le}\label{DistComponentsEqual}
  Let  $\widetilde v_{jj}=2\tr([\vP_a^\top \widetilde\vM(\vP_a,\vSigma^{(j)})\vP_a \vSigma^{(j)}]^2)>0$. Then, under the $j$th local hypothesis $H_0^{(j)}$, the following convergence in distribution holds
\[\widetilde Q_{Nj}\stackrel{\mathcal D}{\to}\widetilde Q_{j}\stackrel{\mathcal D}{=}\frac 1 {\sqrt{\widetilde v_{jj}}}\sum\limits_{i=1}^a {\lambda_i \Upsilon_i}\qquad (j=1,\dots,d)
\]
 with $\Upsilon_1,...,\Upsilon_a\stackrel{i.i.d.}{\sim}\chi_1^2$ and $\lambda_1,...,\lambda_a$ the eigenvalues of $\left(\vSigma^{(j)^{1/2}}\vP_a^\top\widetilde\vM(\vP_a,\vSigma^{(j)})\vP_a\vSigma^{(j)^{1/2}}\right)$.

In addition, we have for $\vZ = (Z_{11},Z_{12},...,Z_{1d},Z_{21},...,Z_{ad})^\top \sim \mathcal{N}_{ad}(\vnull_{ad}, \vSigma)$ under the global null hypothesis $\mathcal H_0$ convergence in distribution 
 $$\widetilde \vQ_N=(\widetilde Q_{N1},...,\widetilde Q_{Nd})^\top\stackrel{\mathcal{D}}{\to}\widetilde\vQ=(\widetilde Q_1,..., \widetilde Q_d)^\top := \left(\widetilde v_{j j}^{-1/2}\cdot  [\vP_a\vZ^{(j)} ]^\top \widetilde\vM(\vP_a,\vSigma^{(j)}) \vP_a\vZ^{(j)} \right)_{j \in \{1,\dots,d\}}$$ with $\vZ^{(j)} := (Z_{1j},...,Z_{aj})^\top$. 
Moreover, if $\mathcal T \subset \{1,\dots,d\}$ denotes the indices of the true local hypotheses, we have
    $(  \widetilde Q_{Nj})_{j\in \mathcal T} \stackrel{\mathcal D}{\to}
    (\widetilde Q_j)_{j \in \mathcal{T}}.$
\end{Le}
The condition $\widetilde v_{jj}>0$ ensures that $\Var(\widetilde Q_{j})=1$, see  \cite{Mathai1992}. This is, e.g., satisfied if $\vP_a^\top\widetilde\vM(\vP_a,\vSigma^{(j)})\vP_a\neq\vnull_{a\times a}$ and $\Var(X_{ij1})>0$ for each $i=1,...,a$ resulting in $\vSigma^{(j)}>0$. In what follows we assume that $\widetilde v_{jj}>0$
for all $j=1,...,d$. 
Moreover, it is important to note that the covariance matrix of the joint limit vector $\widetilde \vQ$, denoted as $\vR$, is not necessarily diagonal due to potential dependencies between different components from the same group. The concrete form of 
$\vR$ is given in equation \eqref{eq:CovarianceMatrix} of the supplement \eqref{eq:CovarianceMatrix}.

 In the case of a WTS, the limit distribution of the components 
 $\widetilde Q_{Nj}$ simplifies to $\chi_{a-1}^2$. However, as demonstrated by \cite{dickhaus2015} using straightforward and intuitive examples, vastly different distributions can result in the same type of marginal distribution, even when sharing the same dependence structure $\vR$. In this case, \cite{dickhaus2015} derive the distribution of the limit vector $\widetilde\vQ$, called multivariate chi-square distribution of generalized Wishart-type, and show that it is in general not possible to calculate the corresponding quantiles. 
Nevertheless, it is reasonable to work with the standardized vector $\widetilde\vQ_N$, and it is important to mention that its marginal limit distributions are known, but usually the distribution of $\widetilde\vQ$ is not. However, this leads to major differences compared to most existing multiple contrast testing procedures as in \cite{rubarth2022a}, \cite{rubarth2022} or \cite{gunawardana2019}. In particular, 
 existing multiple contrast testing procedures for the one-way (M)ANOVA are 
 simultaneous testing procedures that are
 based on the vector of t-test statistics $\widetilde \vT = (\widetilde T_1,\dots,\widetilde T_{ad})$ with { $\widetilde T_j=\vc_{j\bullet}\vT/\sqrt{\vc_{j\bullet}^\top \vSigma \vc_{j\bullet}}=\sqrt{N}\vc_{j\bullet}(\overline \vX-\vmu)/\sqrt{\vc_{j\bullet}^\top \vSigma \vc_{j\bullet}}$} and $\vC = (\vc_{1\bullet}^\top,...,\vc_{ad\bullet}^\top)^\top$ 
 
 Under the global null hypothesis $\vC\vmu=\vnull$, $\widetilde \vT$ follows an asymptotic multivariate normal distribution and its equicoordinate quantile  (see for example \cite{konietschke2012} for more details) is used as critical value for the local (using the test statistic 
 $|\widetilde T_j|$)
 and global (using the test statistic $\max_j |\widetilde T_j|$) test decision.
 {In contrast to other procedures, such as multiple testing with step-wise adjustments, this allows to obtain simultaneous confidence intervals for the effects $\vc_{j\bullet}^\top \vmu$ by inverting the local tests.} Moreover, this leads to powerful test decisions (see for example \cite{konietschke2013} or \cite{harrar2022}) especially as  
 all marginals have the same standard normal limit distribution.
 In our case with quadratic forms, this is no longer true. Although the limit $\widetilde \vQ$ is standardized, and therefore all components have the same variance, the marginals are in general not equal. In fact, Lemma~\ref{DistComponentsEqual} shows that the distribution depends on the eigenvalues of a complex matrix. 
  Using equicoordinate quantiles from an approximated limit distribution (e.g. through Monte Carlo approximations of the limit) 
   would affect the ability to detect derivations from the null hypothesis, and therefore, its power. Hence, we use a different Monte-Carlo approach that is based on ideas for the bootstrap from \cite{munko2024} to get adequate (empirical) quantiles and, thus, critical values.  
{The concrete algorithm } is described in detail in the supplement. Here, we sketch it shortly:

We generate $\vZ_1^{MC}\sim\mathcal{N}_{ad}(\vnull_{ad},\widehat\vSigma)$. Since $\widetilde\vQ$ can be seen as a function of such random vectors (see the proofs in the supplement for more details) we calculate  $\widetilde\vQ^{MC,1}$ based on $\vZ_1^{MC}/\sqrt{N}$. Repeating this $B$ times (e.g. B=$10,000$) we obtain  $\widetilde\vQ^{MC}=(\widetilde\vQ^{MC,1},...,\widetilde\vQ^{MC,B})$. 
{Denote by $\widetilde q_{j,\beta}^{MC}$ the empirical $(1-\beta$) quantile  for {$\widetilde\vQ_j$}, $j=1,...,d$, based on the above Monte-Carlo realization. Then, we determine an appropriate level $\beta$ to control the family-wise type-I error rate, through
\[\widetilde \beta= \max\left(\beta \in \left\lbrace 0,\frac{1}{B},..., \frac{B-1}{B}\right\rbrace\Big\lvert \frac{1}{B} \sum\limits_{b=1}^B \max\limits_{j=1,...,d}\left(\ind\left(\vQ_j^{MC,b}>q_{j,\beta}^{MC}\right)\right)\leq \alpha\right),\]
i.e., the maximum local level, which results in a global level of $\alpha$. Now, an asymptotic correct test for the global hypotheses can be defined, using these quantiles.}

\begin{Le}\label{LemmaMCEq}
{F}or $\alpha \in(0,1)$ and the empirical quantiles
$\widetilde q_{1,B}^{MC},...,\widetilde q_{d,B}^{MC}$ of $\widetilde\vQ^{MC}$, we have that  
\[\varphi_{\widetilde \vQ}^{MC}(N,B):=\max\limits_{j=1,...,d} \ind\left\{\left(\frac{\widetilde Q_{Nj}}{\widetilde q_{j,B}^{MC}}\right)>1\right\}\]
is an asymptotic level $\alpha$ test for the global hypothesis $\mathcal H_0:\vmu_{1}=...=\vmu_{a}$, i.e., for all sequences $B_N\to\infty$ as $N\to\infty$, we have
$$ \lim_{N\to\infty} \E\left[\varphi_{\widetilde\vQ}^{MC}(N,B_N)\right] = \alpha $$
under $\mathcal H_0$. Moreover, the test controls the family-wise error rate asymptotically in the strong sense, i.e. $$ \lim_{N\to\infty} \E\left[\max\limits_{j\in\mathcal T} \ind\left\{\left(\frac{\widetilde Q_{Nj}}{\widetilde{q}_{j,B_N}^{MC}}\right)>1\right\}\right] \leq \alpha $$
for all index subsets $\mathcal T\subset\{1,\dots,d\}$ of true local hypotheses. 
\end{Le}
The global test $\varphi_{\widetilde \vQ}^{MC}$ implies a multiple contrast test for our multiple testing problem by rejecting the local hypotheses 
$\mathcal H_0^{(j)}$ for which ${\widetilde Q_{Nj}}/{\widetilde q_{j,B}^{MC}}>1$ holds.{This again allows to develop simulataneous confidence regions, which, in this case, are based on confidence ellipsoids.}
\\
Instead of the proposed Monte Carlo approach, resampling techniques can be used to derive quantiles and thus critical values. This will be explained in a more general framework in \Cref{Resampling}. This more general framework will be explained below.

\section{Inferring General MANOVA Hypotheses}\label{general}

In the previous section, we addressed the question of identifying the component(s) where differences occur
in one-way MANOVA. Here, we explain, how to adapt our approach to other testing problems, e.g. to identify which groups differ rather than which components.

Thereto, for each of the $\binom{a}{2}$ pairwise comparisons between two groups, a quadratic form can be defined investigating their mean vectors. In this section we generalize these two testing multiple testing problems (difference in components resp. in groups). To this end, we consider the general linear hypothesis
\begin{align*}
    \mathcal H_0: \mathbf{C}\vmu = \vbeta
\end{align*}
for a hypothesis matrix $\mathbf{C}\in\R^{r\times ad}$ and corresponding vector $\vbeta\in\R^r$ for $r\in\mathbb N$. 
Local hypotheses are obtained by partitioning the hypothesis matrix $\vC = \left(\vC_1^\top,\dots,\vC_L^\top\right)^\top$ into $L\in\mathbb N$ block matrices $\vC_1\in\R^{r_1\times ad},\dots,\vC_L\in\R^{r_L\times ad}$ and the vector $\vbeta = (\vbeta_1^\top,\dots,\vbeta_L^\top)^\top$ into $L$ vectors $\vbeta_1,\dots,\vbeta_L$ with corresponding number of rows $r_\ell$ as
\begin{align*}
    \mathcal H_{0}^{(\ell)}: \mathbf{C}_\ell\vmu = \vbeta_\ell, \quad \ell\in\{1,\dots,L\}.
\end{align*}

This contains the equality of expectation vectors treated in the previous section as special case. Other partitions allow, e.g., to determine the responsible groups for a rejection. Now, we want to formulate an  analogous multiple testing procedure for this more general multiple testing problem. Similar to Section 3, this leads to quadratic forms for the $\ell$-th local hypothesis

\[ Q_{N\ell}={N}\cdot \frac{[\vC_\ell\overline \vX - \vbeta_\ell ]^\top\vM(\vC_\ell,\widetilde\vSigma)[\vC_\ell\overline \vX - \vbeta_\ell ]}{\sqrt{2\tr([\vC_\ell^\top \vM(\vC_\ell,\widetilde\vSigma)\vC_\ell \widetilde\vSigma]^2)}}\cdot\ind\left( \tr\left([\vC_\ell^\top \vM(\vC_\ell,\widetilde\vSigma)\vC_\ell \widetilde\vSigma]^2\right)>0\right).\]

In contrast to $\widetilde \vM$ considered in Section~3, the dimension of the second argument of $\vM$ is slightly different, while all other requirements on the function are still needed. The corresponding asymptotic limits are derived below.

\begin{Le}\label{DistComponentsGeneral}
  Let  $v_{\ell\ell} =2\tr([\vC_\ell^\top \vM(\vC_\ell,\vSigma)\vC_\ell \vSigma]^2)>0$. Then
under the $\ell$-th local hypothesis,  we get
\[ Q_{N\ell}\stackrel{\mathcal D}{\to} Q_\ell\stackrel{\mathcal D}{=}  \frac{1}{\sqrt{v_{\ell\ell}}}\sum\limits_{i=1}^{r_\ell} \lambda_i \Upsilon_i\] with $\Upsilon_1,..., \Upsilon_{r_\ell}\stackrel{i.i.d.}{\sim}\chi_1^2$ and $\lambda_1,...,\lambda_{r_\ell}$ the eigenvalues of $\left(\vSigma^{1/2}\vC_\ell^\top \vM(\vP_a,\vSigma)\vC_\ell\vSigma^{1/2}\right)$.
In addition, with $\vZ \sim \mathcal{N}_{ad}(\vnull_{ad}, \vSigma)$ under the global null hypothesis $\mathcal H_0$ also
$$\vQ_N=(Q_{N1},..., Q_{NL})^\top\stackrel{\mathcal{D}}{\to}\vQ=( Q_1,...,  Q_L)^\top := \left(v_{\ell \ell}^{-1/2}\cdot  [\vC_\ell\vZ ]^\top \vM(\vC_\ell,\vSigma) \vC_\ell\vZ \right)_{\ell \in \{1,\dots,L\}}.$$
Moreover, if $\mathcal T \subset \{1,\dots,L\}$ denotes the indices of true local hypotheses, we have
    $(  Q_{N\ell})_{\ell \in \mathcal T} \stackrel{\mathcal D}{\to}
    ( Q_\ell)_{\ell \in \mathcal{T}}.$
\end{Le}

To fulfill the condition of this Lemma, an additional assumption is required:\vspace{0.1cm}
\begin{itemize}
    \item[(A3)] for each $\ell\in\{1,\dots,L\}$, there is a nonzero eigenvalue of $\vC_\ell^\top \vM(\vC_\ell,\vSigma)\vC_\ell \vSigma$.\vspace{0.1cm}
\end{itemize}

A sufficient condition for (A3) is that $\vSigma$ is positive definite and that $\vC_\ell^\top \vM(\vC_\ell,\vSigma)\vC_\ell \neq \vnull_{ad\times ad}$ holds for all $\ell$. However, this condition is not necessary, as we have seen in Section~\ref{Equality}, and the restrictiveness of this condition strongly depends on the chosen $\vM$. For example, in case of an ATS this results in $\vC_\ell^\top\vC_\ell\neq \vnull_{ad\times ad}$, which is fullfilled for all non-trivial matrices $\vC_\ell$ while for the WTS, among other things,  $\vC_\ell^\top\vSigma\vC_\ell\neq \vnull_{ad\times ad}$ is necessary.

Similar to Section~3, the previous Lemma allows to construct Monte-Carlo realizations $\vQ^{MC}=(\vQ^{MC,1},...,\vQ^{MC,B})$ of $\vQ$ and critical values 
$q_{1,B}^{MC},...,q_{L,B}^{MC}$ based on them. This results in asymptotic correct multiple level $\alpha$ tests.
\begin{Le}\label{LemmaMC}
For $\alpha \in(0,1)$, let $q_{1,B}^{MC},...,q_{L,B}^{MC}$ be the corresponding quantiles of $\vQ^{MC}$. Then 
\[\varphi_{\vQ}^{MC}(N,B):=\max\limits_{\ell=1,...,L} \ind\left\{\left(\frac{ Q_{N\ell}}{q_{\ell,B}^{MC}}\right)>1\right\}\]
is an asymptotic level $\alpha$ test for the global hypothesis $\mathcal H_0:\vC\vmu=\vbeta$.
Moreover, the test controls the family-wise error rate asymptotically in the strong sense. 
\end{Le}

\section{Resampling Procedures }\label{Resampling}
Bootstrap techniques often lead to preferable small sample behaviour regarding type-I error and power, see, for example, \cite{sattler2021}, \cite{segbehoe2022} or \cite{baumeister2024}. For this reason, we study two bootstrap approaches: a parametric and a wild bootstrap. For ease of presentation, we focus on the (in simulations better performing) parametric bootstrap in the paper and discuss the wild bootstrap only in the supplement.

To this end, let $\vX_{i 1}^\star := \widehat\vSigma_i^{1/2}\vY_{i1},...,\vX_{i n_i}^\star:=\widehat\vSigma_i^{1/2}\vY_{in_i}$,   $i=1,...,a$, be the parametric bootstrap random vectors with $\vY_{11},...,\vY_{an_a}\stackrel{i.i.d.}{\sim} \mathcal N_{d}(\vnull_{d},\vI_d)$ independent {of} the realizations.
Note that we still do not postulate a parametric model on our original data. Drawing the bootstrap observations for each group from a multivariate normal distribution is motivated by the asymptotic behaviour of $\mathbf T$, cf. Section~\ref{general}. 

Furthermore, let ${\overline \vX}^\star=({\overline\vX_1^\star}^\top,...,{\overline\vX_a^\star}^\top)^\top$ with ${\overline\vX_i^\star}=n_i^{-1}\sum_{k=1}^{n_i}\vX_{ik}^\star$ denote the parametric bootstrap counterpart of the mean vector and
$\widehat \vSigma^\star=N\cdot\bigoplus_{i=1}^a n_i^{-1} \widehat \vSigma_i^\star$,
where $\widehat \vSigma_i^\star$ denotes the empirical covariance matrix of the $i$th parametric bootstrap sample. 
Then, the parametric bootstrap quadratic forms are given by

\[ Q_{N\ell}^\star={N}\cdot \frac{[\vC_\ell{\overline \vX}^\star ]^\top\vM(\vC_\ell,\widehat\vSigma^\star)[\vC_\ell{\overline \vX}^\star ]}{\sqrt{2\tr([\vC_\ell^\top \vM(\vC_\ell,\widehat\vSigma^\star)\vC_\ell \widehat\vSigma^\star]^2)}}\cdot \ind\left(\tr\left([\vC_\ell^\top \vM(\vC_\ell,\widehat\vSigma^\star)\vC_\ell \widehat\vSigma^\star]^2\right)>0\right)\]
for $\ell\in\{1,\dots,L\}$.

In the special setting of  Section~\ref{Equality},
the parametric bootstrap quadratic forms are

\[\widetilde Q_{Nj}^\star={N}\cdot \frac{[\vP_a{\overline \vX^{(j)}}^\star ]^\top\widetilde\vM(\vP_a,{\widehat \vSigma^{(j)\star}})[\vP_a{\overline \vX^{(j)}}^\star]}{\sqrt{2\tr([\vP_a^\top \widetilde\vM(\vP_a, \widehat\vSigma^{(j)\star})\vP_a \widehat\vSigma^{(j)\star}]^2)} }
\cdot \ind\left( \tr\left([\vP_a^\top \widetilde\vM(\vP_a, \widehat\vSigma^{(j)\star})\vP_a \widehat\vSigma^{(j)\star}]^2\right)>0\right)\]
with component-wise mean vectors 
${\overline \vX^{(j)}}^\star=(\frac{1}{n_1}\sum\limits_{k=1}^{n_1}X_{1kj}^\star,...,\frac{1}{n_a}\sum\limits_{k=1}^{n_a}X_{akj}^\star)$
and sample covariances   ${\widehat\vSigma^{(j)\star}}$. The general asymptotic validity of this resampling approach is stated below.

\begin{theorem}\label{PBTheorem1} If  Assumptions (A1) and (A2) are fulfilled, it holds:
Given the data, the conditional  distribution of
\begin{itemize}

\item[(a)]{ $\sqrt{N}\ {\overline \vX}^\star$ converges weakly to $ \mathcal{N}_{ad}\left(\vnull_{ad},\vSigma\right)$ in probability}.
\item[(b)]  {$\widetilde\vQ_N^\star=(\widetilde Q_{N1}^\star,...,\widetilde Q_{Nd}^\star)^\top$ converges weakly to $\widetilde\vQ$ in probability.}
\item[(c)]  {$\vQ_N^\star=(Q_{N1}^\star,...,Q_{NL}^\star)^\top$ converges weakly to $\vQ$ in probability, if also (A3) is fullfilled.}
\end{itemize}
{Moreover, we have ${\widehat \vSigma^{\star}} \to  \vSigma$ in probability, where 
$\vSigma$ is as in Lemma~\ref{DistComponentsGeneral}.
}
\end{theorem}

Similar to the Monte-Carlo approach from the previous sections, a large number of repetitions, say $B$, leads to realizations $\vQ_N^{\star,1},...,\vQ_N^{\star,B}$ of $\vQ_N^\star$ with corresponding bootstrap quantiles which can be used to define asymptotically valid testing procedures.

\begin{Le}\label{LemmaBS}
For $\alpha \in(0,1)$, let $q_{1,B}^{\star},...,q_{L,B}^{\star}$ be the corresponding quantiles of $\vQ_N^{\star,1},...,\vQ_N^{\star,B}$. Then
\[\varphi_{\vQ}^{\star}(N,B):=\ind\left\{\max\limits_{\ell=1,...,L} \left(\frac{ Q_{N\ell}}{q_{\ell,B}^{\star}}\right)>1\right\}\]
is an asymptotic correct level $\alpha$ test for the global null hypothesis $\mathcal H_0:\vC\vmu=\vbeta$ that controls the family-wise error rate asymptotically in the strong sense.
\end{Le}

To check the hypothesis $\mathcal{H}_0:\vmu_1=...=\vmu_a$  based on $\widetilde \vQ_N$ from Section~\ref{Equality}, 
  the appropriate test would be given through
\[\varphi_{\widetilde\vQ}^{\star}(N,B):=\ind\left\{\max\limits_{j=1,...,d} \left(\frac{ \widetilde Q_{Nj}}{\widetilde q_{j,B}^{\star}}\right)>1\right\},\]
where $\widetilde q_{1,B}^{\star},...,\widetilde q_{d,B}^{\star}$ are the quantiles calculated using realisations $\widetilde \vQ_N^{\star,1},...,\widetilde\vQ_N^{\star,B}$
of $\widetilde \vQ_N^{\star}$.

  Especially in the context of resampling approaches, large dimensions, or large number of groups the computing effort of an approach becomes more relevant. Here, the quadratic form based multiple contrast test (QFMCT) is often preferable over classical MCT or quadratic forms. However, the concrete efficiency gain depends on the chosen quadratic form.
  For the simplest quadratic form, the ATS without standardization 
  using $\widetilde\vM(\vP_a,\vSigma^{(j)})=\vI_a$, it is easy to calculate and compare the computational complexity. In the case of classical multiple contrast tests with a Tukey-type matrix, it is $\Lan(a^3d^3)$ and for this  ATS, even $\Lan(a^3d^3+a^2d^2)$. For the quadratic forms based multiple contrast test, of course, $d$ quadratic forms are calculated, but each of them with a substantially smaller vector. This leads to a computational complexity of $\Lan(da^3+da^2+d^2)$, which is considerably smaller.
 Applying recent results from \cite{sattler2025}, this can even be reduced.\\
  For the standardized ATS or the WTS, the quadratic forms based multiple contrast test is even more preferable, since for the required covariance matrix estimation and matrix multiplications, the complexity is growing exponentially with the dimension. This makes our approach attractive from a computational perspective.

\section{Simulations}\label{Simulations}

The previous sections established the asymptotic validity of various QFMCTs. In this section, we assess their performance in finite samples through simulations. 
From \cite{sattler2023a}, it follows in the case of only two groups that for the ATS and WTS, using a hypothesis matrix with only one row would lead to the same values of the test statistic. Following this,  the value of the classical MCT and the QFMCT
coincides. Also, this does not influence the previous results, it makes it reasonable to focus here on $a>2$, and we choose $a=4$. Moreover, we consider unbalanced group sizes $n_1=n_2= N/3$ and $n_3=n_4= N/6$ with $N \in\{30, 60, 120 \}$ as we are especially interested in the performance for small to moderate sample sizes. Throughout, we consider $4$-dimensional observation vectors generated independently according to the model $\vX_{ik} = \vmu_i + \vSigma^{1/2} \vZ_{ik}, i=1,\dots, a, k=1,\dots,n_i$ and error terms $\vZ_{ik} = (Z_{ik1},...,Z_{ik5})^\top$ based on independently generated $Z_{ikj}, i=1,\dots, a, k=1,\dots,n_i, j=1,\dots,4,$ with the following distributions:\vspace{0.2cm}

\begin{itemize}

\item a standard normal distribution, i.e. $Z_{ikj} \sim  \mathcal {N}(0,1).$\vspace{0.1cm}
\item a student-t distribution with 9 degrees of freedom, i.e. $\sqrt{7/9}\cdot Z_{ikj}\sim t_9$.\vspace{0.2cm}
\end{itemize} 
We choose autoregressive matrices $(\vSigma_1)_{j_1 j_2}=(\vSigma_2)_{j_1 j_2}=0.65^{|j_1-j_2|}$ and heterogenous compound symmetry 
 matrices $\vSigma_3=\vSigma_4=\diag(3,4,5,6)+\veins_4\veins_4^\top$ to have not only heterogeneous $\vSigma_i$ but also $\vSigma^{(j)}$ and under the respective null hypothesis we know for the ATS, $Q_\ell\sim \sum_{s=1}^{r_\ell} \lambda_s B_s$, with $B_s\stackrel{i.i.d.}{\sim}\chi_1^2$ and $\lambda_1,...,\lambda_{r_\ell}$ the eigenvalues of $(\vC_\ell\vSigma\vC_\ell^\top)$. Therefore, despite their standardization, we have different distributions of the quadratic forms, as mentioned in Section~\ref{Equality}. For this reason, usage of the same quantile for all components of $\widetilde \vQ$ makes less sense, and the approach of \cite{munko2024} is required.

Since we are particularly interested in the power of the tests, we consider a shift-alternative, given through $\vmu_1=\vmu_2=\vmu_3=\vnull_4$ and $\vmu_4=\delta\cdot\veins_4$. We report the rejection rate in percentages, and print the type-I error rate in bold, if it is within the $95\%$ binomial interval $[0.0458, 0.0543]$.
We use 1,000 bootstrap steps for our parametric bootstrap, 10,000 simulation steps for the Monte-Carlo approach and 10,000 runs for all tests to get reliable results.\\

\subsection{Global Hypothesis}
Here, we focus on the one-way-MANOVA with global hypothesis $\mathcal H_0:\vmu_{1}=...=\vmu_{a}$, which has two natural decompositions in local hypotheses.
On the one hand, the focus on the components, which was the basic idea of chapter \ref{Equality}, which we here denote as $\mathcal H_0^A$. On the other hand, we could focus on identifying the groups between which differences occur. Here we have 6 local hypotheses defined through
\[\vC_1=(1,-1,0,0)\otimes \vI_4\quad \quad\vC_2=(1,0,-1,0)\otimes \vI_4\quad\quad \vC_3=(1,0,0,-1)\otimes \vI_4\]
\[\vC_4=(0,1,-1,0)\otimes \vI_4\quad \quad\vC_5=(0,1,0,-1)\otimes \vI_4\quad\quad \vC_6=(0,0,1,-1)\otimes \vI_4\] 

with $\vbeta_1=...=\vbeta_6=\vnull_4$. We denote this global hypothesis as $\mathcal H_0^B$.
It is clear that this difference is only important for the QFMCT, while other ones are unaffected.

For testing the global hypotheses, besides the QFMCT with parametric bootstrap and Monte-Carlo approach, we simulated the classical ATS with parametric bootstrap ($\psi_{ATS}^\star$).
Moreover, we considered the classical MCT using the Tukey-type matrix, with equicoordinate quantiles and based on the parametric bootstrap ($\varphi$ and $\varphi^{\star}$).\\

\captionsetup[table]{aboveskip = 5pt}
\captionsetup[table]{belowskip = 1pt}

  \begin{table}[htbp]
\begin{center}
\begin{scriptsize}
\begin{tabular}{ |l|r|r|r|r|r|r|r|r|r|r|} \hline
\rowcolor{ashgrey}  & {\hspace{-0.1cm} N\hspace{0.01cm}} & {\hspace{-0.1cm} $\delta=0$\hspace{-0.1cm}}&{\hspace{-0.1cm}  $\delta=0.25$\hspace{-0.1cm}}&{\hspace{-0.1cm}  $\delta=0.5$\hspace{-0.1cm}}&{ \hspace{-0.1cm}$\delta=0.75$\hspace{-0.1cm}}&{\hspace{-0.1cm}  $\delta=1.00$\hspace{-0.1cm}}& {\hspace{-0.1cm} $\delta=1.25$\hspace{-0.1cm}}& {\hspace{-0.1cm} $\delta=1.5$\hspace{-0.1cm}}&{ \hspace{-0.1cm} $\delta=1.75$\hspace{-0.1cm}}& {\hspace{-0.1cm} $\delta=2.00$\hspace{-0.1cm}}\\\hline 
 \rowcolor{lightyellow}{$\varphi$} & 30 & 30.38 & 31.37 & 35.15 & 40.75 & 47.90 & 56.84 & 66.03 & 74.12 & 81.89 \\ [0.03cm] 
\rowcolor{gainsboro}  {$\varphi^{\star}$} & 30 & 6.01 & 6.13 & 7.86 & 9.68 & 13.49 & 18.50 & 24.63 & 31.67 & 39.27 \\ [0.03cm] 
 \rowcolor{lightyellow} {$\varphi_{\vQ}^{MC}(ATS)$} & 30 & 15.89 & 16.90 & 19.79 & 25.17 & 32.42 & 41.48 & 50.68 & 61.58 & 71.19 \\ [0.03cm] 
\rowcolor{gainsboro}  $\varphi_{\vQ}^\star(ATS)$ & 30 & 6.46 & 6.71 & 8.41 & 10.74 & 14.90 & 21.12 & 27.77 & 35.77 & 44.55 \\ [0.03cm] 
 \rowcolor{lightyellow} $\psi_{ATS}^\star$ & 30 & 2.66 & 2.87 & 4.36 & 6.71 & 11.54 & 17.73 & 26.17 & 36.92 & 49.16 \\[0.03cm]   \hline
\rowcolor{gainsboro}   {$\varphi$} & 60 & 14.78 & 17.16 & 23.94 & 33.74 & 47.46 & 63.68 & 77.58 & 88.02 & 94.46 \\ [0.03cm] 
 \rowcolor{lightyellow} {$\varphi^{\star}$} & 60 & 5.66 & 6.25 & 9.78 & 15.67 & 26.24 & 39.01 & 55.24 & 69.99 & 82.14 \\ [0.03cm] 
\rowcolor{gainsboro}  {$\varphi_{\vQ}^{MC}(ATS)$} & 60 & 9.85 & 11.19 & 16.93 & 25.90 & 40.61 & 55.80 & 71.69 & 84.33 & 92.68 \\ [0.03cm] 
 \rowcolor{lightyellow}$\varphi_{\vQ}^\star(ATS)$ & 60 & 5.55 & 6.31 & 10.21 & 16.36 & 28.25 & 42.22 & 58.41 & 73.74 & 85.80 \\ [0.03cm] 
\rowcolor{gainsboro}  $\psi_{ATS}^\star$ & 60 & 3.85 & 4.81 & 9.54 & 17.76 & 33.13 & 50.33 & 68.81 & 83.92 & 93.02 \\ [0.03cm]  \hline
 \rowcolor{lightyellow}  {$\varphi$} & 120 & 9.43 & 12.75 & 24.26 & 43.86 & 67.74 & 86.05 & 95.87 & 99.31 & 99.87 \\ [0.03cm] 
\rowcolor{gainsboro}  {$\varphi^{\star}$} & 120 & \bf{5.41} & 7.32 & 15.99 & 32.24 & 55.76 & 77.69 & 91.78 & 98.05 & 99.64 \\ [0.03cm]  
 \rowcolor{lightyellow} {$\varphi_{\vQ}^{MC}(ATS)$} & 120 & 6.95 & 10.39 & 20.73 & 40.35 & 64.67 & 83.84 & 94.90 & 99.00 & 99.87 \\ [0.03cm] 
\rowcolor{gainsboro}  $\varphi_{\vQ}^\star(ATS)$ & 120 & \bf{4.95} & 7.90 & 16.32 & 34.06 & 57.67 & 79.04 & 92.55 & 98.50 & 99.73 \\ [0.03cm] 
 \rowcolor{lightyellow} $\psi_{ATS}^\star$ & 120 & 4.35 & 7.77 & 19.28 & 40.93 & 69.42 & 87.94 & 97.37 & 99.57 & 99.95 \\ [0.03cm]  
 \hline
\end{tabular}
\end{scriptsize}
\end{center}
\caption{Power of global tests for $\Hypo^A$ under a shift-alternative for 4 unbalanced groups
with  $n_1=n_2= N/3$ and $n_3=n_4= N/6$.
The 4-dimensional error terms are based on the standard normal distribution and covariance matrices $(\vSigma_1)_{j_1 j_2}=(\vSigma_2)_{j_1 j_2}=0.65^{|j_1-j_2|}$ respective $\vSigma_3=\vSigma_4=\diag(3,4,5,6)+\veins_4\veins_4^\top$.}
\label{tab:CSA1}
\end{table}

In \Cref{tab:CSA1} and \Cref{tab:CSA2}, it can be readily seen that the type-I error rate for the bootstrap version of the MCT and QFMCT is comparably good while the bootstrap ATS has a worse small sample approximation, but overall still good values. The power under the shift-alternative is consistently higher for the QFMCT, which, especially for smaller sample sizes, results in differences up to 0.05.
  \begin{table}[htbp]
\begin{center}
\begin{scriptsize}
\begin{tabular}{ |l|r|r|r|r|r|r|r|r|r|r|} \hline
\rowcolor{ashgrey} & {\hspace{-0.1cm} N\hspace{0.01cm}} & {\hspace{-0.1cm} $\delta=0$\hspace{-0.1cm}}&{\hspace{-0.1cm}  $\delta=0.25$\hspace{-0.1cm}}&{\hspace{-0.1cm}  $\delta=0.5$\hspace{-0.1cm}}&{ \hspace{-0.1cm}$\delta=0.75$\hspace{-0.1cm}}&{\hspace{-0.1cm}  $\delta=1.00$\hspace{-0.1cm}}& {\hspace{-0.1cm} $\delta=1.25$\hspace{-0.1cm}}& {\hspace{-0.1cm} $\delta=1.5$\hspace{-0.1cm}}&{ \hspace{-0.1cm} $\delta=1.75$\hspace{-0.1cm}}& {\hspace{-0.1cm} $\delta=2.00$\hspace{-0.1cm}}\\\hline 
 \rowcolor{lightyellow} {$\varphi$} & 30 & 28.14 & 29.26 & 33.44 & 39.88 & 49.15 & 58.19 & 67.34 & 76.52 & 83.91 \\ [0.03cm] 
 \rowcolor{gainsboro} {$\varphi^{\star}$} & 30 & \bf{4.83} & 5.25 & 6.76 & 9.71 & 13.39 & 18.94 & 24.80 & 33.19 & 42.53 \\ [0.03cm] 
  \rowcolor{lightyellow} {$\varphi_{\vQ}^{MC}(ATS)$} & 30 & 14.53 & 15.64 & 19.02 & 24.64 & 33.37 & 42.77 & 52.40 & 63.78 & 73.38 \\ [0.03cm] 
\rowcolor{gainsboro}  $\varphi_{\vQ}^\star(ATS)$ & 30 & 5.73 & 5.94 & 7.59 & 10.70 & 15.62 & 21.46 & 28.33 & 38.29 & 48.17 \\ [0.03cm] 
  \rowcolor{lightyellow} $\psi_{ATS}^\star$ & 30 & 2.07 & 2.57 & 4.05 & 7.00 & 11.98 & 18.00 & 26.62 & 38.22 & 50.42 \\ [0.03cm] \hline 
 \rowcolor{gainsboro}  {$\varphi$} & 60 & 13.78 & 15.92 & 22.92 & 33.88 & 49.88 & 65.06 & 79.42 & 89.69 & 95.48 \\ [0.03cm] 
   \rowcolor{lightyellow}{$\varphi^{\star}$} & 60 & \bf{4.77} & 5.77 & 9.29 & 16.38 & 27.60 & 41.31 & 57.12 & 72.84 & 84.62 \\ [0.03cm] 
\rowcolor{gainsboro}  {$\varphi_{\vQ}^{MC}(ATS)$} & 60 & 8.58 & 10.72 & 16.48 & 26.78 & 41.10 & 56.50 & 72.43 & 85.23 & 93.29 \\ [0.03cm] 
 \rowcolor{lightyellow}  $\varphi_{\vQ}^\star(ATS)$ & 60 & \bf{4.61} & 6.04 & 9.71 & 16.79 & 28.51 & 42.31 & 59.01 & 74.87 & 86.02 \\ [0.03cm] 
\rowcolor{gainsboro}  $\psi_{ATS}^\star$ & 60 & 3.30 & 4.66 & 8.84 & 18.13 & 32.43 & 49.82 & 67.79 & 82.88 & 92.18 \\ [0.03cm]  \hline
 \rowcolor{lightyellow}   {$\varphi$} & 120 & 8.90 & 12.67 & 24.45 & 45.92 & 68.71 & 86.98 & 96.21 & 99.30 & 99.89 \\ [0.03cm] 
\rowcolor{gainsboro}  {$\varphi^{\star}$} & 120 & \bf{4.94} & 7.49 & 15.89 & 33.69 & 56.98 & 79.63 & 92.66 & 98.38 & 99.77 \\ [0.03cm] 
 \rowcolor{lightyellow}  {$\varphi_{\vQ}^{MC}(ATS)$} & 120 & 7.25 & 10.01 & 21.17 & 41.49 & 64.19 & 84.56 & 95.34 & 98.99 & 99.83 \\ [0.03cm] 
\rowcolor{gainsboro}  $\varphi_{\vQ}^\star(ATS)$ & 120 & \bf{5.04} & 7.41 & 16.43 & 34.58 & 57.05 & 80.10 & 92.95 & 98.30 & 99.70 \\ [0.03cm] 
 \rowcolor{lightyellow}  $\psi_{ATS}^\star$ & 120 & 4.19 & 7.42 & 18.87 & 41.14 & 66.78 & 87.73 & 96.88 & 99.49 & 99.98 \\ [0.03cm] 
 \hline
\end{tabular}
\end{scriptsize}
\end{center}
\caption{Power of global tests for $\Hypo^A$ under a shift-alternative for 4 unbalanced groups
with  $n_1=n_2= N/3$ and $n_3=n_4= N/6$.
The 4-dimensional error terms are based on the $t_9$ distribution and  covariance matrices $(\vSigma_1)_{j_1 j_2}=(\vSigma_2)_{j_1 j_2}=0.65^{|j_1-j_2|}$ respective $\vSigma_3=\vSigma_4=\diag(3,4,5,6)+\veins_4\veins_4^\top$.}
\label{tab:CSA2}
\end{table}
\newpage
In comparison, the QFMCTs in \Cref{tab:GSA1} and \Cref{tab:GSA2} are a bit more conservative, but overall their type-I error rate is similar. Thus, referring to that new decomposition in local hypotheses is favourable.

 \begin{table}[htbp]
\begin{center}
\begin{scriptsize}

\begin{tabular}{ |l|r|r|r|r|r|r|r|r|r|r|} \hline
\rowcolor{ashgrey} & {\hspace{-0.1cm} N\hspace{0.01cm}} & {\hspace{-0.1cm} $\delta=0$\hspace{-0.1cm}}&{\hspace{-0.1cm}  $\delta=0.25$\hspace{-0.1cm}}&{\hspace{-0.1cm}  $\delta=0.5$\hspace{-0.1cm}}&{ \hspace{-0.1cm}$\delta=0.75$\hspace{-0.1cm}}&{\hspace{-0.1cm}  $\delta=1.00$\hspace{-0.1cm}}& {\hspace{-0.1cm} $\delta=1.25$\hspace{-0.1cm}}& {\hspace{-0.1cm} $\delta=1.5$\hspace{-0.1cm}}&{ \hspace{-0.1cm} $\delta=1.75$\hspace{-0.1cm}}& {\hspace{-0.1cm} $\delta=2.00$\hspace{-0.1cm}}\\\hline 
 \rowcolor{lightyellow} {$\varphi$} & 30 & 30.38 & 31.37 & 35.15 & 40.75 & 47.90 & 56.84 & 66.03 & 74.12 & 81.89 \\ [0.03cm] 
\rowcolor{gainsboro}  {$\varphi^{\star}$} & 30 & 6.01 & 6.13 & 7.86 & 9.68 & 13.49 & 18.50 & 24.63 & 31.67 & 39.27 \\ [0.03cm] 
 \rowcolor{lightyellow}  {$\varphi_{\vQ}^{MC}(ATS)$} & 30 & 7.97 & 8.58 & 11.24 & 15.39 & 22.02 & 30.64 & 41.32 & 53.79 & 64.97 \\ [0.03cm] 
 \rowcolor{gainsboro} $\varphi_{\vQ}^\star(ATS)$ & 30 & 4.06 & 4.14 & 5.73 & 7.16 & 11.80 & 17.75 & 24.23 & 33.97 & 45.23 \\ [0.03cm] 
 \rowcolor{lightyellow}  $\psi_{ATS}^\star$ & 30 & 2.66 & 2.87 & 4.36 & 6.71 & 11.54 & 17.73 & 26.17 & 36.92 & 49.16 \\ [0.03cm]  \hline
\rowcolor{gainsboro}    {$\varphi$} & 60 & 14.78 & 17.16 & 23.94 & 33.74 & 47.46 & 63.68 & 77.58 & 88.02 & 94.46 \\ [0.03cm] 
 \rowcolor{lightyellow}  {$\varphi^{\star}$} & 60 & 5.66 & 6.25 & 9.78 & 15.67 & 26.24 & 39.01 & 55.24 & 69.99 & 82.14 \\ [0.03cm] 
 \rowcolor{gainsboro} {$\varphi_{\vQ}^{MC}(ATS)$} & 60 & 6.49 & 8.13 & 13.92 & 23.88 & 39.74 & 57.49 & 74.84 & 87.54 & 94.69 \\ [0.03cm] 
 \rowcolor{lightyellow}  $\varphi_{\vQ}^\star(ATS)$ & 60 & \bf{4.79} & 6.29 & 10.86 & 19.38 & 33.64 & 50.47 & 68.40 & 82.73 & 91.97 \\ [0.03cm] 
  \rowcolor{gainsboro}$\psi_{ATS}^\star$ & 60 & 3.85 & 4.81 & 9.54 & 17.76 & 33.13 & 50.33 & 68.81 & 83.92 & 93.02 \\ [0.03cm] \hline 
 \rowcolor{lightyellow}  {$\varphi$} & 120 & 8.90 & 12.67 & 24.45 & 45.92 & 68.71 & 86.98 & 96.21 & 99.30 & 99.89 \\ [0.03cm] 
  \rowcolor{gainsboro}{$\varphi^{\star}$} & 120 & \bf{4.94} & 7.49 & 15.89 & 33.69 & 56.98 & 79.63 & 92.66 & 98.38 & 99.77 \\ [0.03cm] 
 \rowcolor{lightyellow}  {$\varphi_{\vQ}^{MC}(ATS)$} & 120 & 5.81 & 9.26 & 21.34 & 44.41 & 72.49 & 90.01 & 97.79 & 99.64 & 99.96 \\ [0.03cm] 
 \rowcolor{gainsboro} $\varphi_{\vQ}^\star(ATS)$ & 120 & \bf{5.38} & 8.28 & 20.07 & 41.99 & 70.35 & 88.69 & 97.30 & 99.56 & 99.93 \\ [0.03cm] 
  \rowcolor{lightyellow} $\psi_{ATS}^\star$ & 120 & 4.35 & 7.77 & 19.28 & 40.93 & 69.42 & 87.94 & 97.37 & 99.57 & 99.95 \\ [0.03cm] 
 \hline
\end{tabular}
\end{scriptsize}
\end{center}
\caption{Power of global tests for $\Hypo^B$ under a shift-alternative for 4 unbalanced groups
with  $n_1=n_2= N/3$ and $n_3=n_4= N/6$.
The 4-dimensional error terms are based on the standard normal distribution and covariance matrices $(\vSigma_1)_{j_1 j_2}=(\vSigma_2)_{j_1 j_2}=0.65^{|j_1-j_2|}$ respective  $\vSigma_3=\vSigma_4=\diag(3,4,5,6)+\veins_4\veins_4^\top$.}
\label{tab:GSA1}
\end{table}

For both distributions, it is evident that while the classical ATS, using the same bootstrap approach, is too conservative, it sometimes achieves higher power compared to the multiple contrast tests. In general, although all multiple contrast tests seem to perform a bit better for the $t_9$ distribution, the effect is small, and the comparison of the results suggests that the performance does not depend on the chosen distribution.

  \begin{table}[htbp]
\begin{center}
\begin{scriptsize}
\begin{tabular}{ |l|r|r|r|r|r|r|r|r|r|r|} \hline
\rowcolor{ashgrey} & {\hspace{-0.1cm} N\hspace{0.01cm}} & {\hspace{-0.1cm} $\delta=0$\hspace{-0.1cm}}&{\hspace{-0.1cm}  $\delta=0.25$\hspace{-0.1cm}}&{\hspace{-0.1cm}  $\delta=0.5$\hspace{-0.1cm}}&{ \hspace{-0.1cm}$\delta=0.75$\hspace{-0.1cm}}&{\hspace{-0.1cm}  $\delta=1.00$\hspace{-0.1cm}}& {\hspace{-0.1cm} $\delta=1.25$\hspace{-0.1cm}}& {\hspace{-0.1cm} $\delta=1.5$\hspace{-0.1cm}}&{ \hspace{-0.1cm} $\delta=1.75$\hspace{-0.1cm}}& {\hspace{-0.1cm} $\delta=2.00$\hspace{-0.1cm}}\\\hline 
 \rowcolor{lightyellow} {$\varphi$} & 30 & 28.14 & 29.26 & 33.44 & 39.88 & 49.15 & 58.19 & 67.34 & 76.52 & 83.91 \\ [0.03cm] 
\rowcolor{gainsboro}  {$\varphi^{\star}$} & 30 & \bf{4.83} & 5.25 & 6.76 & 9.71 & 13.39 & 18.94 & 24.80 & 33.19 & 42.53 \\ [0.03cm] 
  \rowcolor{lightyellow} {$\varphi_{\vQ}^{MC}(ATS)$} & 30 & 7.09 & 7.92 & 10.56 & 15.32 & 22.57 & 31.91 & 42.01 & 54.50 & 65.82 \\ [0.03cm] 
 \rowcolor{gainsboro} $\varphi_{\vQ}^\star(ATS)$ & 30 & 3.40 & 3.91 & 5.00 & 7.90 & 12.63 & 18.25 & 25.49 & 35.08 & 45.84 \\ [0.03cm] 
  \rowcolor{lightyellow} $\psi_{ATS}^\star$ & 30 & 2.07 & 2.57 & 4.05 & 7.00 & 11.98 & 18.00 & 26.62 & 38.22 & 50.42 \\ [0.03cm]  \hline
 \rowcolor{gainsboro}  {$\varphi$} & 60 & 13.78 & 15.92 & 22.92 & 33.88 & 49.88 & 65.06 & 79.42 & 89.69 & 95.48 \\ [0.03cm] 
 \rowcolor{lightyellow}  {$\varphi^{\star}$} & 60 & \bf{4.77} & 5.77 & 9.29 & 16.38 & 27.60 & 41.31 & 57.12 & 72.84 & 84.62 \\ [0.03cm] 
\rowcolor{gainsboro}  {$\varphi_{\vQ}^{MC}(ATS)$} & 60 & 6.24 & 7.41 & 13.05 & 24.08 & 40.02 & 57.00 & 74.51 & 87.14 & 93.91 \\ [0.03cm] 
  \rowcolor{lightyellow} $\varphi_{\vQ}^\star(ATS)$ & 60 & 4.52 & 5.67 & 10.25 & 19.39 & 33.30 & 49.84 & 67.74 & 81.51 & 90.86 \\ [0.03cm] 
 \rowcolor{gainsboro} $\psi_{ATS}^\star$ & 60 & 3.30 & \bf{4.66} & 8.84 & 18.13 & 32.43 & 49.82 & 67.79 & 82.88 & 92.18 \\ [0.03cm]  \hline
 \rowcolor{lightyellow}   {$\varphi$} & 120 & 8.90 & 12.67 & 24.45 & 45.92 & 68.71 & 86.98 & 96.21 & 99.30 & 99.89 \\ [0.03cm] 
 \rowcolor{gainsboro} {$\varphi^{\star}$} & 120 & \bf{4.94} & 7.49 & 15.89 & 33.69 & 56.98 & 79.63 & 92.66 & 98.38 & 99.77 \\ [0.03cm] 
 \rowcolor{lightyellow}  {$\varphi_{\vQ}^{MC}(ATS)$} & 120 & 5.53 & 9.00 & 21.40 & 44.69 & 71.04 & 89.41 & 97.45 & 99.61 & 99.98 \\ [0.03cm] 
 \rowcolor{gainsboro} $\varphi_{\vQ}^\star(ATS)$ & 120 & \bf{4.98} & 8.05 & 19.78 & 42.22 & 68.65 & 88.29 & 96.94 & 99.50 & 99.97 \\ [0.03cm] 
 \rowcolor{lightyellow}  $\psi_{ATS}^\star$ & 120 & 4.19 & 7.42 & 18.87 & 41.14 & 66.78 & 87.73 & 96.88 & 99.49 & 99.98 \\ [0.03cm] 
 \hline
\end{tabular}
\end{scriptsize}
\end{center}
\caption{Power of global tests for $\Hypo^B$ under a shift-alternative for 4 unbalanced groups
with  $n_1=n_2= N/3$ and $n_3=n_4= N/6$.
The 4-dimensional error terms are based on the $t_9$ distribution and covariance matrices $(\vSigma_1)_{j_1 j_2}=(\vSigma_2)_{j_1 j_2}=0.65^{|j_1-j_2|}$ respective $\vSigma_3=\vSigma_4=\diag(3,4,5,6)+\veins_4\veins_4^\top$.}
\label{tab:GSA2}
\end{table}

When comparing the results for $\Hypo^A$ with $\Hypo^B$, for the parametric bootstrap, the decomposition in 6 local hypotheses clearly leads to a better power, with differences up to 0.11. For the Monte-Carlo approach this is vice versa, although the results are too liberal for both decompositions.\\   This shows that although it is the same global hypothesis, an adequate choice of local alternatives increases the probability of rejecting the hypothesis of interest. Moreover, this demonstrates that the ability of the QFMCT to detect deviation from the null hypothesis does not require that all local hypotheses are false. While for $\Hypo^A$ all four local hypotheses are violated under the alternative, for $\Hypo^B$ this is only the case for half of them.\\

Finally, the Monte Carlo and asymptotic approaches in this simulation need larger sample sizes to exhibit good results. Again, the QFMCT is favourable through it good results for moderate sample sizes. Moreover, its computation (using 10,000 runs) is even faster compared to the corresponding approach based upon equicoordinate quantiles, which is computationally quite demanding.

\subsection{Local Test}
Since the QFMCT, as a side effect, also allows us to test local hypotheses, we investigate its power in this regard in this section. To this aim, we consider ANOVA-type statistics based on a parametric bootstrap for each of our local hypotheses, adjusted with the approach of \cite{holm1979} (denoted with $BH$), as it is less restrictive than the classical Bonferroni adjustment. For the same two distributions as earlier, we assess two power properties: we measure the percentage of all violated local hypotheses are correctly rejected (simultaneous power), as well as the proportion in which at least one of the incorrect ones is rejected (any power).

The results for $\Hypo^A$ are less informative, since all local hypotheses are violated, and therefore, they partially coincide with values from the earlier section. Hence, we focus on $\Hypo^B$, while the results for $\Hypo^A$ can be found in the supplement.

 \begin{table}[htbp]
\begin{center}
\begin{scriptsize}
\begin{tabular}{ |l|r|r|r|r|r|r|r|r|r|} \hline
\rowcolor{ashgrey} & {\hspace{-0.1cm} N\hspace{0.01cm}} &{\hspace{-0.1cm}  $\delta=0.25$\hspace{-0.1cm}}&{\hspace{-0.1cm}  $\delta=0.5$\hspace{-0.1cm}}&{ \hspace{-0.1cm}$\delta=0.75$\hspace{-0.1cm}}&{\hspace{-0.1cm}  $\delta=1.00$\hspace{-0.1cm}}& {\hspace{-0.1cm} $\delta=1.25$\hspace{-0.1cm}}& {\hspace{-0.1cm} $\delta=1.5$\hspace{-0.1cm}}&{ \hspace{-0.1cm} $\delta=1.75$\hspace{-0.1cm}}& {\hspace{-0.1cm} $\delta=2.00$\hspace{-0.1cm}}\\\hline

\rowcolor{gainsboro}   {$\varphi_{\vQ}^{MC}(ATS)$}-Any & 30  & 4.83 & 7.55 & 11.94 & 18.90 & 28.41 & 39.36 & 52.34 & 64.04 \\ [0.03cm] 
 \rowcolor{lightyellow}  {$\varphi_{\vQ}^\star(ATS)$}-Any & 30  & 1.90 & 3.52 & 5.45 & 9.81 & 16.28 & 22.99 & 32.85 & 44.37 \\ [0.03cm] 
\rowcolor{gainsboro}   $\psi_{ATS}^\star(BH)$-Any & 30 & 3.07 & 5.63 & 8.93 & 14.45 & 22.90 & 32.74 & 43.87 & 55.91 \\ 

 \rowcolor{lightyellow}  {$\varphi_{\vQ}^{MC}(ATS)$}-Sim & 30 & 0.29 & 0.78 & 1.28 & 3.09 & 5.91 & 9.55 & 16.11 & 24.24 \\ [0.03cm] 
 \rowcolor{gainsboro}  {$\varphi_{\vQ}^\star(ATS)$}-Sim & 30 & 0.08 & 0.24 & 0.54 & 1.31 & 2.45 & 4.53 & 8.07 & 12.82 \\[0.03cm]
   \rowcolor{lightyellow}  $\psi_{ATS}^\star(BH)$-Sim & 30  & 0.06 & 0.14 & 0.37 & 0.95 & 1.79 & 3.63 & 6.49 & 10.39 \\ [0.03cm] \hline
 
 \rowcolor{gainsboro}  {$\varphi_{\vQ}^{MC}(ATS)$}-Any & 60 & 4.95 & 11.06 & 21.52 & 38.18 & 56.50 & 74.29 & 87.28 & 94.61 \\ [0.03cm] 
 \rowcolor{lightyellow}  {$\varphi_{\vQ}^\star(ATS)$}-Any & 60 & 3.81 & 8.62 & 17.48 & 32.34 & 49.46 & 67.90 & 82.49 & 91.88 \\ [0.03cm] 
 \rowcolor{gainsboro}  $\psi_{ATS}^\star(BH)$-Any & 60  & 5.49 & 11.95 & 23.27 & 39.92 & 58.63 & 75.96 & 88.59 & 95.04 \\ [0.03cm] 

\rowcolor{lightyellow}   {$\varphi_{\vQ}^{MC}(ATS)$}-Sim & 60 & 0.41 & 1.19 & 3.44 & 9.22 & 17.74 & 32.22 & 48.35 & 65.48 \\ [0.03cm] 
\rowcolor{gainsboro}   {$\varphi_{\vQ}^\star(ATS)$}-Sim & 60 & 0.31 & 0.98 & 2.62 & 7.62 & 14.41 & 27.15 & 42.49 & 59.74 \\ [0.03cm] 
 \rowcolor{lightyellow}  $\psi_{ATS}^\star(BH)$-Sim & 60  & 0.24 & 0.76 & 2.21 & 6.50 & 12.97 & 24.84 & 39.60 & 56.66 \\ [0.03cm] \hline

 \rowcolor{gainsboro}  {$\varphi_{\vQ}^{MC}(ATS)$}-Any & 120 & 6.62 & 19.24 & 43.21 & 71.92 & 89.84 & 97.76 & 99.64 & 99.96 \\ [0.03cm] 
 \rowcolor{lightyellow}  {$\varphi_{\vQ}^\star(ATS)$}-Any & 120 & 5.92 & 18.20 & 40.89 & 69.83 & 88.51 & 97.28 & 99.56 & 99.93 \\ [0.03cm] 
 \rowcolor{gainsboro}   $\psi_{ATS}^\star(BH)$-Any & 120 & 8.98 & 23.71 & 49.43 & 77.36 & 92.50 & 98.55 & 99.75 & 100.00 \\ [0.03cm] 

 \rowcolor{lightyellow}  {$\varphi_{\vQ}^{MC}(ATS)$}-Sim & 120 & 0.59 & 3.36 & 10.92 & 27.70 & 51.65 & 73.78 & 89.00 & 96.30 \\ [0.03cm] 
 \rowcolor{gainsboro}  {$\varphi_{\vQ}^\star(ATS)$}-Sim & 120 & 0.48 & 2.94 & 9.99 & 26.26 & 49.73 & 72.42 & 88.18 & 95.90 \\ [0.03cm] 
 \rowcolor{lightyellow}  $\psi_{ATS}^\star(BH)$-Sim & 120 & 0.46 & 2.75 & 9.31 & 24.85 & 48.24 & 70.37 & 86.92 & 95.26 \\ [0.03cm] 
 \hline
\end{tabular}
\end{scriptsize}
\end{center}
\caption{Simultaneous power and any power of local tests for $\Hypo^B$ under a shift-alternative for 4 unbalanced groups
with  $n_1=n_2= N/3$ and $n_3=n_4= N/6$.
The 4-dimensional error terms are based on the standard normal distribution and covariance matrices $(\vSigma_1)_{j_1 j_2}=(\vSigma_2)_{j_1 j_2}=0.65^{|j_1-j_2|}$ respective   $\vSigma_3=\vSigma_4=\diag(3,4,5,6)+\veins_4\veins_4^\top$.}
\label{tab:GSAl1}
\end{table}

  \begin{table}[htbp]
\begin{center}
\begin{scriptsize}
\begin{tabular}{ |l|r|r|r|r|r|r|r|r|r|} \hline
\rowcolor{ashgrey} & {\hspace{-0.1cm} N\hspace{0.01cm}} & {\hspace{-0.1cm}  $\delta=0.25$\hspace{-0.1cm}}&{\hspace{-0.1cm}  $\delta=0.5$\hspace{-0.1cm}}&{ \hspace{-0.1cm}$\delta=0.75$\hspace{-0.1cm}}&{\hspace{-0.1cm}  $\delta=1.00$\hspace{-0.1cm}}& {\hspace{-0.1cm} $\delta=1.25$\hspace{-0.1cm}}& {\hspace{-0.1cm} $\delta=1.5$\hspace{-0.1cm}}&{ \hspace{-0.1cm} $\delta=1.75$\hspace{-0.1cm}}& {\hspace{-0.1cm} $\delta=2.00$\hspace{-0.1cm}}\\\hline 

 \rowcolor{gainsboro}  {$\varphi_{\vQ}^{MC}(ATS)$}-Any & 30 &  4.43 & 7.18 & 12.01 & 19.83 & 29.50 & 40.35 & 53.12 & 64.87 \\ [0.03cm] 
\rowcolor{lightyellow}   {$\varphi_{\vQ}^\star(ATS)$}-Any & 30 &  2.00 & 3.16 & 5.89 & 11.04 & 16.67 & 24.30 & 33.94 & 44.95 \\ [0.03cm] 
\rowcolor{gainsboro}$\psi_{ATS}^\star(BH)$-Any & 30  & 3.20 & 5.02 & 8.72 & 15.24 & 23.78 & 32.85 & 44.33 & 56.36 \\ [0.03cm] 

\rowcolor{lightyellow}{$\varphi_{\vQ}^{MC}(ATS)$}-Sim& 30 & 0.39 &  0.68 &   1.43 &   3.19 &   5.71 &  10.03 &  16.15 &  25.23\\ [0.03cm]
  \rowcolor{gainsboro}  {$\varphi_{\vQ}^\star(ATS)$}-Sim & 30  & 0.01 & 0.20 & 0.45 & 1.23 & 2.38 & 4.41 & 7.77 & 13.50 \\ [0.03cm] 
    \rowcolor{lightyellow} $\psi_{ATS}^\star(BH)$-Sim & 30  & 0.01 & 0.14 & 0.32 & 0.85 & 1.84 & 3.59 & 6.22 & 11.53 \\ [0.03cm]\hline

\rowcolor{gainsboro}   {$\varphi_{\vQ}^{MC}(ATS)$}-Any & 60  & 4.45 & 10.39 & 21.75 & 38.53 & 56.01 & 73.97 & 86.86 & 93.86 \\ [0.03cm] 
\rowcolor{lightyellow}   {$\varphi_{\vQ}^\star(ATS)$}-Any & 60  & 3.36 & 8.13 & 17.59 & 32.11 & 48.93 & 67.19 & 81.28 & 90.77 \\ [0.03cm] 
\rowcolor{gainsboro}  $\psi_{ATS}^\star(BH)$-Any & 60 &  5.11 & 11.66 & 23.77 & 40.86 & 58.06 & 75.62 & 87.58 & 94.45 \\ [0.03cm] 

\rowcolor{lightyellow}   {$\varphi_{\vQ}^{MC}(ATS)$}-Sim & 60  & 0.30 & 1.15 & 3.64 & 9.47 & 18.24 & 31.09 & 48.14 & 63.77 \\ [0.03cm] 
 \rowcolor{gainsboro}  {$\varphi_{\vQ}^\star(ATS)$}-Sim & 60  & 0.28 & 0.79 & 2.66 & 7.34 & 14.89 & 26.24 & 42.34 & 58.18 \\ [0.03cm]
   \rowcolor{lightyellow} $\psi_{ATS}^\star(BH)$-Sim & 60  & 0.18 & 0.68 & 2.30 & 6.21 & 13.19 & 24.38 & 39.58 & 55.39 \\ [0.03cm] \hline

\rowcolor{gainsboro}   {$\varphi_{\vQ}^{MC}(ATS)$}-Any & 120  & 6.51 & 19.35 & 43.65 & 70.59 & 89.32 & 97.44 & 99.59 & 99.98 \\ [0.03cm] 
\rowcolor{lightyellow}   {$\varphi_{\vQ}^\star(ATS)$}-Any & 120  & 5.77 & 17.84 & 41.28 & 68.12 & 88.22 & 96.93 & 99.48 & 99.97 \\ [0.03cm] 
\rowcolor{gainsboro}  $\psi_{ATS}^\star(BH)$-Any & 120 & 8.91 & 23.76 & 49.69 & 75.59 & 92.16 & 98.27 & 99.71 & 99.99 \\ [0.03cm] 

 \rowcolor{lightyellow} {$\varphi_{\vQ}^{MC}(ATS)$}-Sim & 120  & 0.58 & 2.83 & 10.44 & 27.72 & 51.97 & 72.94 & 88.23 & 95.76 \\ [0.03cm] 
  \rowcolor{gainsboro}{$\varphi_{\vQ}^\star(ATS)$}-Sim & 120  & 0.46 & 2.57 & 9.83 & 26.03 & 50.04 & 71.17 & 87.47 & 95.14 \\ [0.03cm] 
     \rowcolor{lightyellow}$\psi_{ATS}^\star(BH)$-Sim & 120  & 0.37 & 2.30 & 9.05 & 24.84 & 48.15 & 69.39 & 85.94 & 94.42 \\ [0.03cm] 
 \hline
\end{tabular}
\end{scriptsize}
\end{center}
\caption{Simultaneous power and any power of local tests for $\Hypo^B$ under a shift-alternative for 4 unbalanced groups
with  $n_1=n_2= N/3$ and $n_3=n_4= N/6$.
The 4-dimensional error terms are based on the $t_9$ distribution and covariance matrices $(\vSigma_1)_{j_1 j_2}=(\vSigma_2)_{j_1 j_2}=0.65^{|j_1-j_2|}$ respective   $\vSigma_3=\vSigma_4=\diag(3,4,5,6)+\veins_4\veins_4^\top$.}
\label{tab:GSAl2}
\end{table}

For both considered distributions, the ATS with a bootstrap detects more frequently if at least one of the local hypotheses is violated compared to the QFMCT with parametric bootstrap. Here, the largest difference occurs for $N=30$. However, {$\varphi_{\vQ}^{MC}(ATS)$} is consistently better in simultaneous power. Nevertheless, it depends on the situation which ability is more important and, therefore, which procedure is preferable. Due to its liberal behaviour, the QFMTC using a Monte-Carlo approach should only be used for larger sample sizes. 
In such cases, however, it performs well in detecting both all violated local hypotheses and at least one.

\section{Illustrative Data Analysis}\label{Example}
To illustrate the application of the method, the EEG data set from the \textsc{R}-package \textit{manova.rm} by \cite{Friedrich2019} is considered closer. This study from \cite{staffen2014} was conducted at the University Clinic of Salzburg (Department of Neurology), where electroencephalography (EEG) data were measured from 160 patients. All participants are diagnosed with different diagnoses of impairments, namely  Alzheimer's disease (AD), mild cognitive impairment (MCI), and subjective cognitive complaints (SCC). 
In \Cref{tab:EEG1}, the number of patients divided by sex and diagnosis can be found.

\begin{table}[h]
\centering
\begin{small}
\caption{Number of observations for the different factor level combinations of sex and diagnosis.}
\label{tab:EEG1}
\begin{tabular}{l|c|c|c|}
&AD&MCI&SCC\\
\hline
male&12&27&20\\
\hline
female&24&30&47\\
\hline
\end{tabular}
\end{small}
\end{table}

MANOVA-based comparisons were already made in \cite{bathke2018}, and we now want to complete this with our new approach.
The original study also further differentiated between subjective cognitive complaints with minimal cognitive dysfunction (SCC+) and without \mbox{(SCC-)}. However, due to sample sizes, this was not done in the data set provided through the package.

For each participant,  three different electrode positions (frontal, temporal, and central) were used together with two kinds of measurements (z-score for brain rate and Hjorth complexity), resulting in observation vector's dimension $d=6$. In \cite{sattler2022}, homogeneity of covariance matrices between different diagnoses as well as different sexes were investigated for this data set. Therefore, we have six groups with heterogeneous covariance matrices and unbalanced sample sizes, making many test procedures incapable.

In \cite{bathke2018}, equality of expectation vectors was rejected, and therefore an influence of the diagnosis was proven. However,
it was never identified where the differences occur. Since all previous analyses of this data set could neither verify an influence of the location nor measurement, identifying the responsible component is, here, not the main issue. Therefore, we will focus on identifying the responsible groups
and determine three local hypotheses
\[\Hypo^{(1)}:\vmu_{AD}=\vmu_{MCI}\quad\quad \Hypo^{(2)}:\vmu_{AD}=\vmu_{SCC}\quad\quad \Hypo^{(3)}:\vmu_{MCI}=\vmu_{SCC}\]
which together build the global hypothesis $\Hypo: \vmu_{AD}=\vmu_{MCI}=\vmu_{SCC}$.
Since we have vectors of dimension 6, each of these local hypotheses consists of 6 sub-hypotheses, one for each component of the vector. So, for the classical MCT we have 18 single hypotheses, in partitions of 6 for each local hypothesis. The test is conducted with 5.000 bootstrap runs and for both genders separately.

\captionsetup[table]{position=below,skip=.4cm}
\begin{table}[htp]
\centering
\begin{small}
\begin{tabular}{x{1.3pt}lx{1.3pt}r|r|r|r|r|rx{1.3pt}cx{1.3pt}}\specialrule{1.3pt}{0pt}{0pt}
&\multicolumn{6}{cx{1.3pt}}{p-values of $\varphi^\star$ for 6 sub-hypotheses}&p-value of $\varphi_{\vQ}^\star(ATS)$\\
\specialrule{1.3pt}{0pt}{0pt}
$\Hypo^{(1)}$-Male &   74.14 & 74.08 & 68.60 & 67.38 & 87.96 & 61.28& 31.04 \\ \hline
  $\Hypo^{(2)}$-Male & 1.46 & 3.32 & 0.84 & 5.88 & 27.02 & 20.10 & 0.98 \\   \hline 
  $\Hypo^{(3)}$-Male & 2.10 & 2.98 & 0.40 & 2.04 & 32.16 & 0.58& 0.04 \\  \specialrule{1.3pt}{0pt}{0pt}
  $\Hypo^{(1)}$-Female &  97.92 & 100.00 & 99.46 & 78.22 & 99.96 & 87.58& 75.54 \\  \hline 
  $\Hypo^{(2)}$-Female & 1.14 & 8.86 & 3.16 & 2.80 & 38.40 & 7.08& 0.38 \\  \hline 
  $\Hypo^{(3)}$-Female & 2.74 & 3.08 & 2.04 & 9.52 & 59.10 & 7.66& 0.38 \\ 
\specialrule{1.3pt}{0pt}{0pt}\end{tabular}
\caption{p-values in percent for QFMTC and classical MCT, both based on parametric bootstrap and investigating local hypotheses  $\Hypo^{(1)}:\vmu_{AD}=\vmu_{MCI}$, $\Hypo^{(2)}:\vmu_{AD}=\vmu_{SCC}$ and $\Hypo^{(3)}:\vmu_{MCI}=\vmu_{SCC}$. }\label{EEGResultate}
\end{small}
\end{table}

In \Cref{EEGResultate}, the results can be found, where it contains for each local hypothesis 6  p-values when considering $\varphi^\star$ and only one for $\varphi_{\vQ}^\star(ATS)$, while all of them are calculated according to \cite{munko2024a}. 
It is not surprising that for both genders, the global hypotheses can be rejected for the global level $\alpha=5\%$ since this was also a result of \cite{bathke2018}. Furthermore, all multiple contrast tests determine the differences between AD and SCC as well as MCI and SCC. The classic MCT can also determine the components responsible for this and reject $\Hypo^{(2)}$ and $\Hypo^{(3)}$ for men also to the smaller global level $\alpha=1\%$. In contrast, all 4 rejections of QFMCT hold for this level, again showing QFMCT's advantage by identifying even relatively small deviations from the null hypothesis. 

\section{Conclusion}\label{Conclussion}
In the present work, we introduce a novel approach that combines two well-established concepts, multiple contrast tests and quadratic forms, to leverage the strengths of both. This leads to quadratic form based multiple contrast tests (QFMCT), that can be defined for quite general quadratic forms and contrasts. Employing an appropriate bootstrap approach, this results in a consistent multiple contrast test for general MANOVA models. Thereby, we first introduce the approach for equality of expectations with a focus on the components. Later, we generalize it to larger classes of hypotheses. Our simulations indicate that while the type-I error rate is comparable to that of the bootstrap version of existing multiple contrast tests, our approach offers improved power. Additionally, it generally requires fewer computations, making it preferable for inferring general linear hypotheses regarding expectation vectors in multiple groups, in particular for large numbers of groups.\\ 

The current paper is focused on expectation vectors as parameters of interest since they are the most common effect size. However, the technique of a quadratic form based multiple contrast test can also be applied in other contexts. For example, we also envision QFMCTs for a variety of parameter vectors like quantiles, relative effects, vectorized covariance matrices and many more. In future work, we will develop corresponding QFMCTs and 
investigate their large and small sample properties.

\section*{Acknowledgements}
Merle Munko gratefully acknowledges the support of the Deutsche Forschungsgemeinschaft (grant no. GRK 2297 MathCoRe).

\bibliographystyle{apalike}
\bibliographystyle{unsrtnat}

\addcontentsline{toc}{chapter}{Bibliography}
\bibliography{Literaturnew}
\newpage
\appendix
\section{Proofs}
We will here only conduct the proofs for the results for \Cref{general} and \Cref{Resampling}, since the results from \Cref{Equality} can be seen as a special case thereof. This follows since there exist a matrix $\vA_j\in\R^{a\times ad}$ fulfilling $\overline \vX^{(j)}=\vA_j \overline \vX$.
Hence, by setting $\vC_j=\vP_a\vA_j$ and $\vbeta_j=\vnull_{a}$, all statements from \Cref{Equality} are direct consequences from \Cref{general}.

\begin{proof}[Proof of Lemma~\ref{DistComponentsGeneral}]\label{ProofQn}
With the central limit theorem it holds $\sqrt{N}(\overline \vX-\vmu)\stackrel{\mathcal{D}}{\to}\vZ\sim\mathcal{N}_{ad}(\vnull_{ad},\vSigma)$.
Let $\mathcal T \subset \{1,\dots, L\}$ denote the index subset of true local null hypotheses. From the consistency of $\widetilde \vSigma$ it follows $\widehat v_{\ell\ell}=2\tr([\vC_\ell^\top \vM(\vC_\ell,\widetilde\vSigma)\vC_\ell \widetilde\vSigma]^2)\stackrel{\mathcal P}{\to} v_{\ell\ell}$.
With continuous mapping  and Slutzky's theorem, therefore 

\[\sqrt{N}\cdot\left(
  \widehat v_{\ell\ell}^{-1/2}\cdot  \vC_\ell\overline \vX - \vbeta_\ell,
    \vM(\vC_\ell,\widehat\vSigma)(\vC_\ell\overline \vX - \vbeta_\ell)
\right)_{\ell\in\mathcal T} \stackrel{\mathcal{D}}{\to}
    \left(v_{\ell\ell}^{-1/2}\cdot\vC_\ell\vZ, \vM(\vC_\ell,\vSigma)\vC_\ell\vZ\right)_{\ell\in\mathcal T}=:\vUpsilon.\]
Now consider the continuous function
\[\vf:\R^{2\sum_{\ell\in\mathcal T}r_\ell}\to\R^{|\mathcal T|}, (\vx_1,\vy_1,...,\vx_d,\vy_d)\mapsto(\vy_1^\top\vx_1,...,\vy_d^\top\vx_d)^\top.\] 
Together with the continuous mapping theorem, it holds
\begin{align*}
    (  Q_{N\ell})_{\ell \in \mathcal T}&=N\cdot\left(\widehat v_{\ell\ell}^{-1/2}\cdot(\vC_\ell\overline \vX -\vbeta_\ell)^\top \vM(\vC_\ell,\widehat\vSigma)(\vC_\ell \overline \vX - \vbeta_\ell)\right)_{\ell\in\mathcal T}
    \\&= \vf\left(\sqrt{N}\cdot\left(\widehat v_{\ell\ell}^{-1/2}\cdot
   (\vC_\ell\overline \vX - \vbeta_\ell),
    \vM(\vC_\ell,\widehat\vSigma)(\vC_\ell\overline \vX - \vbeta_\ell)
\right)_{\ell\in\mathcal T}\right)
    \stackrel{\mathcal{D}}{\to} \vf(\vUpsilon) = (  Q_{\ell})_{\ell \in \mathcal T}.
\end{align*}

The distribution of the components then follows by the representation theorem for quadratic forms (see e.g. \cite{Mathai1992}).
    
\end{proof}

With
$\sqrt{N}(\overline \vX-\vmu)\stackrel{\mathcal{D}}{\to}\vZ\sim\mathcal{N}_{ad}(\vnull_{ad},\vSigma)$ it holds\\\\
$\begin{array}{ll}( Q_{N\ell}, Q_{N\ell_2})
&\stackrel{\mathcal{D}}{\to} \left([\vC_{\ell_1} \vZ]^\top \vM(\vC_{\ell_1},\vSigma)[\vC_{\ell_1} \vZ], [\vC_{\ell_2} \vZ]^\top \vM(\vC_{\ell_2},\vSigma)[\vC_{\ell_2} \vZ]\right)\cdot (v_{\ell_1\ell_1}v_{\ell_2\ell_2})^{-1/2}
\\&=
\left(\vZ^\top\vC_{\ell_1}^\top \vM(\vC_{\ell_1},\vSigma)\vC_{\ell_1} \vZ, \vZ^\top\vC_{\ell_2}^\top \vM(\vC_{\ell_2},\vSigma)\vC_{\ell_2} \vZ\right)\cdot (v_{\ell_1\ell_1}v_{\ell_2\ell_2})^{-1/2}
\end{array}$\\\\

if the local null hypotheses $\mathcal H_0^{(\ell_1)}$ and $\mathcal H_0^{(\ell_2)}$ are true, which fulfills the situation from Theorem 3.2d.4 in \cite{Mathai1992}. Therefore
we know
\begin{equation}\label{eq:CovarianceMatrix}  
\vR_{\ell_1,\ell_2}:=\Cov(Q_{\ell_1},Q_{\ell_2})=\frac{2\tr\left(\vC_{\ell_2}^\top \vM(\vC_{\ell_2},\vSigma)\vC_{\ell_2}\vSigma\vC_{\ell_2}^\top \vM(\vC_{\ell_2},\vSigma)\vC_{\ell_2}\vSigma\right)}{\sqrt{v_{\ell_1\ell_1}\cdot v_{\ell_2\ell_2}}}
\end{equation}

\paragraph{Monte-Carlo-approach}\label{MC-approach}
{For the Monte-Carlo- approach, we proceed as follows}

\begin{enumerate} \item Based on the data, calculate the empirical covariance matrix $\widehat\vSigma$.\vspace{0.1cm}
    \item Generate $\vZ_1^{MC}\sim\mathcal{N}_{ad}(\vnull_{ad},\widehat\vSigma)$.\vspace{0.1cm}
    \item Calculate the Monte-Carlo quadratic form 
    $\vQ^{MC,1} :=\begin{pmatrix}\widehat v_{11}^{-1/2}\cdot (\vC_1\vZ_1^{MC})^\top \vM(\vC_1,\widehat\vSigma)\vC_1 \vZ_1^{MC}
    \\
    \vdots\\
   \widehat v_{LL}^{-1/2}\cdot (\vC_L\vZ_1^{MC})^\top \vM(\vC_L,\widehat\vSigma)\vC_L\vZ_1^{MC}
    \end{pmatrix}$
     with $\widehat{v}_{\ell\ell}:= 2\tr([\vC_\ell^\top \vM(\vC_\ell,\widehat\vSigma)\vC_\ell \widehat\vSigma]^2)$. \vspace{0.1cm}
  
    \item Repeat steps 2. and 3. a large number of times, e.g. B=10.000 times, to obtain $\vQ^{MC,1},...,\vQ^{MC,B}$.\\

\end{enumerate}

{Then \Cref{DistComponentsGeneral} ensures that the random variables $\vQ^{MC,1},...,\vQ^{MC,B}$ are asymptotically distributed as the asymptotic distribution of the test statistic under the null hypothesis, using the consistency of $\widehat{\vSigma}$.
Thereto, it remains to calculate adequate quantiles based on them, which is done by adapting the technique from \cite{munko2024} developed for bootstrap to our Monte-Carlo approach. Denote by $q_{\ell,\beta}^{MC}$ the empirical $(1-\beta$) quantile  for {$|\vQ_\ell|$}, $\ell=1,...,p$, based on the above Monte-Carlo realization. Then we determine an appropriate level $\beta$ to control the family-wise type-I error rate through
\[\widetilde \beta= \max\left(\beta \in \left\lbrace 0,\frac{1}{B},..., \frac{B-1}{B}\right\rbrace\Big\lvert \frac{1}{B} \sum\limits_{b=1}^B \max\limits_{\ell=1,...,L}\left(\ind\left(\vQ_\ell^{MC,b}>q_{\ell,\beta}^{MC}\right)\right)\leq \alpha\right),\]
i.e., the maximum local level, which results in a global level of $\alpha$. The resulting 
method is asymptotically valid by Lemma~S8 in \cite{munko2024}.}

\begin{proof}[Proof of \Cref{PBTheorem1}]
First we define $\overline \vX_{i}^\star:=\frac{1}{n_i} \sum\limits_{k=1}^{n_i}\vX_{ik}^\star$ for $i=1,...,a$ and based on this $\overline \vX^\star=(\overline \vX_{1}^\star,...,\overline \vX_{a}^\star)^\top$ as bootstrap pendant to $\overline \vX_i$ resp. $\overline \vX$. Now we use the conditional Lindeberg-Feller-theorem to show that given the data (which we denote through $|\vX$) it holds $\overline \vX_i^\star\stackrel{\mathcal{D}}{\to} \vY_i\sim \mathcal{N}_a(\vnull_a,\vSigma_i)$.
From the generation of the bootstrap observations we know that given the data, $\vX_{i1}^\star,...,\vX_{in_i}^\star$ are independent with\\\\
$\begin{array}{lll}
1.)&&\sum\limits_{k=1}^{n_i}\E\left(\frac{\sqrt{N}}{n_i}\vX_{ik}^\star|\vX\right)=\sum\limits_{k=1}^{n_i}\frac{\sqrt{N}}{n_i}\E\left(\vX_{ik}^\star|\vX\right)=\vnull_d\\[1.7ex]
2.)&&\sum\limits_{k=1}^{n_i}\Cov\left(\frac{\sqrt{N}}{n_i}\vX_{ik}^\star|\vX\right)=\frac{N}{n_i^2}\sum\limits_{k=1}^{n_i}\widehat \vSigma_i\stackrel{\mathcal P}{\to}\frac{1}{\kappa_i}\vSigma_i\\[1.7ex]
3.)&&
\sum\limits_{k=1}^{n_i}\E\left(\lvert\lvert\frac{\sqrt{N}}{n_i}\vX_{ik}^\star\lvert \lvert^2 \cdot\ind\left({\lvert\lvert\frac{\sqrt{N}}{n_i}\vX_{ik}^\star\lvert \lvert> \delta}\right)\Big \lvert \vX\right)\\
&=&
{n_i^2}\sum\limits_{k=1}^{n_i}\E\left(\lvert\lvert\vX_{ik}^\star\lvert \lvert^2 \cdot\ind\left({\lvert\lvert\vX_{ik}^\star\lvert \lvert> \delta \cdot {n_i}/\sqrt{N}}\right)\Big \lvert \vX\right)\\
&\leq& \frac 1 {\kappa_i}
\sum\limits_{k=1}^{n_i}\sqrt{\E\left(\lvert\lvert\vX_{ik}^\star\lvert \lvert^4|\vX\right)}\cdot\sqrt{\E\left( \ind\left({\lvert\lvert\vX_{ik}^\star\lvert \lvert> \delta\cdot  {n_i}/\sqrt{N}}\right) \Big \lvert \vX\right)}\stackrel{P}{\rightarrow}0,
\end{array}$\\\\
where Cauchy-Bunjakowski-Schwarz inequality is used in the last line. Together with the fact that, given the data, the observations follow a multivariate normal distribution, the first expectation value is bounded. From (A1) we know $n_i/N\to \kappa_i$ and therefore $ \delta\cdot n_i/\sqrt{N}\to \infty$ so it follows with the conditional Markov inequality
\[P(||\vX_{i1}^\star||>\delta\cdot n_i/\sqrt{N}|\vX)\leq \frac{\sqrt{N}}{\delta\cdot n_i}\cdot \E(||\vX_{i1}^\star||\ |\vX)\leq \frac{\sqrt{N}}{\delta \cdot n_i}\cdot ||(\widehat \vSigma_i)^{1/2}||_{\max}\cdot\E(||\vY||).\] 
with  $\vY\sim\mathcal{N}_d(\vnull_d,\vI_d)$.
Here, we use the max norm defined as  $||\vA||_{\max}=\max\limits_{k_1=1,...,d_1}\left(\max\limits_{k_2=1,...,d_2}\left( |A_{k_1k_2}|\right)\right)$ for a matrix $\vA\in\R^{d_1\times d_2}$.
 Using $\E(||\vY||)<\infty$ and $||(\widehat \vSigma_i)^{1/2}||_{\max}\stackrel{\mathcal P}{\to} ||(\vSigma_i)^{1/2}||_{\max}<\infty$ this together with Slutzky's theorem completes the Lindeberg-Feller conditions. Now, from the independence of groups, it directly follows, that, given the data, the conditional distribution of
$\sqrt{N}\overline\vX^\star$ converges weakly to $\vZ\sim\mathcal{N}_{ad}(\vnull_{ad},\vSigma)$. \\

For the consistency of the covariance estimator, we use again the representation $\vX_{ij}^\star = \widehat{\vSigma}_i^{1/2}\vY_{ij}$ for $\vY_{11},...,\vY_{an_a}\sim\mathcal{N}_d(\vnull_d,\vI_d)$ i.i.d. and independent of the data. Then, we have
\begin{align*}
    \widehat\vSigma_i^\star &= \frac{1}{n_i-1}\sum_{j=1}^{n_i} (\vX_{ij}^\star-\overline{\vX}_{i}^\star)(\vX_{ij}^\star-\overline{\vX}_{i}^\star)^\top
    =\widehat{\vSigma}_i^{1/2} \left( \frac{1}{n_i-1}\sum_{j=1}^{n_i} (\vY_{ij}-\overline{\vY}_{i})(\vY_{ij}-\overline{\vY}_{i})^\top \right) \widehat{\vSigma}_i^{1/2}.
\end{align*}
The term in the middle is the empirical covariance matrix of $\vY_{i1},...,\vY_{in_i}$ which converges in probability to $\vI_d$. Hence, it follows $\widehat\vSigma_i^\star\stackrel{\mathcal{P}}{\to}\vSigma_i$.
Hence, it follows $\widehat\vSigma^\star = N\bigoplus_{i=1}^an_i^{-1}\widehat\vSigma_i^\star\stackrel{\mathcal{P}}{\to}\vSigma$.

Now as continuous function of consistent estimators, ${{2\tr([\vC_\ell^\top \vM(\vC_\ell,\widehat \vSigma^\star)\vC_\ell \widehat \vSigma^\star]^2)}}$ is also a consistent estimator of $v_{\ell\ell}$. Therefore  all necessary properties are shown, and with Slutzky's theorem, the result follows since $\vQ_N^\star$ can be written by using the continuous function $\vf$ and the estimator ${\widehat\vSigma^{\star}}$.\end{proof}

\paragraph{Wild bootstrap} For the wild bootstrap approach let $W_{ik}$, $i=1,...,a,$ be i.i.d. with $\E(W_{ik})=0$, $\Var(W_{ik})=1$ and $\E(W_{ik}^4)<\infty$, while usually distributions are standard normal or Mammen (\cite{Mammen1993}). Then the wild bootstrap observation vectors be given through $\vX_{ik}^\dagger=W_{ik}(\vX_{ik}-\overline{\vX}_i)$ and we define
\[{\overline \vX}^\dagger=\left(\frac{1}{n_1}\sum\limits_{k=1}^{n_1}{\vX_{1k}^{\dagger}}^\top,...,\frac{1}{n_a}\sum\limits_{k=1}^{n_a}{\vX_{ak}^\dagger}^\top\right)^\top.\] Analogues to the parametric bootstrap, based on our wild bootstrap we define  
\[Q_{N\ell}^\dagger={N}\cdot\frac{[\vC_\ell{\overline \vX}^\dagger ]^\top\vM(\vC_\ell,{\widehat \vSigma^{\dagger}})[\vC_\ell{\overline \vX}^\dagger]}{\sqrt{2\tr\left(\left[\vC_\ell^\top \vM(\vC_\ell,{\widehat \vSigma^{\dagger}})\vC_\ell{\widehat \vSigma^{\dagger}}\right]^2\right)}},\quad\ell\in\{1,\dots,L\}.\] 
 using the covariance estimators  $\widehat\vSigma^{\dagger}$ and $\widetilde Q_{Nj}^\dagger$ based on ${\widehat\vSigma^{(j)\dagger}}$.

   \begin{theorem}\label{WBTheorem1}
If Assumptions (A1) and (A2) are fulfilled, it holds: Given the data, the conditional distribution of
\begin{itemize}
\item[(a)]{ $\sqrt{N}\ {\overline \vX}^\dagger$ converges weakly to $ \mathcal{N}_{a}\left(\vnull_{a},\vSigma\right)$ in probability}.
\item[(b)]  {$\widetilde\vQ_N^\dagger=(\widetilde Q_{N1}^\dagger,...,\widetilde Q_{Nd}^\dagger)^\top$ converges weakly to $\widetilde\vQ$ in probability.}
\item[(c)]  {$\vQ_N^\dagger=(Q_{N1}^\dagger,...,Q_{NL}^\dagger)^\top$ converges weakly to $\vQ$ in probability, if also (A3) is fulfilled.}
\end{itemize}
{Moreover, we have $\widehat \vSigma^{\dagger} \stackrel{\mathcal{P}}{\to} \vSigma$ and the unknown covariance matrix can be estimated through this estimator.}
\end{theorem}  
\begin{proof}

For the wild bootstrap approach, we follow the same steps as for the parametric bootstrap.\\
$\begin{array}{lll}
1.)&&\sum\limits_{k=1}^{n_i}\E\left(\frac{\sqrt{N}}{n_i}\vX_{ik}^\dagger|\vX\right)=\sum\limits_{k=1}^{n_i}\frac{\sqrt{N}}{n_i}\E(W_{ik})(\vX_{ik}-\overline{\vX}_i)=\vnull_d\\[1.7ex]
2.)&&\sum\limits_{k=1}^{n_i}\Cov\left(\frac{\sqrt{N}}{n_i}\vX_{ik}^\dagger|\vX\right)\\&=&\frac{N}{n_i^2}\sum\limits_{k=1}^{n_i}
\E((\vX_{ik}^\dagger)(\vX_{ik}^\dagger)^\top|\vX)\\
&=&\frac{N}{n_i^2}\sum\limits_{k=1}^{n_i}\E(W_{ik}^2)\cdot (\vX_{ik}-\overline \vX_{i})(\vX_{ik}-\overline \vX_{i})^\top\\&=&\frac{N(n_i-1)}{n_i^2}  \frac{1}{n_i-1}\sum\limits_{k=1}^{n_i}1\cdot (\vX_{ik}-\overline \vX_{i})(\vX_{ik}-\overline \vX_{i})^\top
\\&=&\frac{N(n_i-1)}{n_i^2}\widehat \Sigma_i\stackrel{\mathcal P}{\to}\frac{1}{\kappa_i}\vSigma_i\\[1.7ex]
\end{array}\\
\begin{array}{lll}3.)&&\lim\limits_{N\to \infty} \sum\limits_{k=1}^{n_i}\E\left(\lvert\lvert\frac{\sqrt{N}}{n_i}\vX_{ik}^\star\lvert \lvert^2 \cdot\ind\left({\lvert\lvert\frac{\sqrt{N}}{n_i}\vX_{ik}^\star\lvert \lvert> \delta}\right)\Big \lvert \vX\right)\\
&=&\lim\limits_{N\to \infty} \frac{{N}}{n_i^2}\sum\limits_{k=1}^{n_i}\E\left(\lvert\lvert\vX_{ik}^\dagger\lvert \lvert^2 \cdot\ind\left({\lvert\lvert\vX_{ik}^\dagger \lvert \lvert> \delta\cdot {n_i}/\sqrt{N}}\right)\Big \lvert \vX\right)\\
&\leq& \frac 1 {\kappa_i}\lim\limits_{N\to \infty} \sum\limits_{k=1}^{n_i}\sqrt{\E\left(\lvert\lvert\vX_{ik}^\dagger\lvert \lvert^4|\vX\right)}\cdot\sqrt{\E\left(\ind\left({\lvert\lvert\vX_{ik}^\dagger\lvert \lvert> \delta\cdot {n_i}/\sqrt{N}}\right)\Big \lvert \vX\right)}\stackrel{\mathcal P}{=}0.
\end{array}$\\\\
Since, given the data, $\vX_{ik}^\dagger$ has finite fourth moments, all remaining steps can be done analogously to the proof of \Cref{PBTheorem1}.
\end{proof}

\begin{Le}\label{LemmaWB}
For $\alpha \in(0,1)$, let $q_{1,B}^{\dagger},...,q_{L,B}^{\dagger}$ be the corresponding quantiles of $\vQ_N^{\dagger}$ according to \cite{munko2024}. Then
\[\varphi_{\vQ}^{\dagger}(N,B):=\ind \left\{\max\limits_{\ell=1,...,L} \left(\frac{ Q_{N\ell}}{q_{\ell,B}^{\dagger}}\right)>1\right\}\]
is an asymptotic correct level $\alpha$ test for the global hypothesis $\mathcal H_0:\vC\vmu=\vbeta$ that controls the family-wise error rate asymptotically in the strong sense.
\end{Le}

\paragraph{Proof of the asymptotic correctness and family-wise error rate control of the tests}
In this paragraph, we prove \Cref{LemmaMC}, \Cref{LemmaBS}, and \Cref{LemmaWB}.
Therefore, we aim to apply Lemma~S8 in the supplement of \cite{munko2024} with $\varepsilon_N := 1/B_N$ and $F:\R^L\to[0,1]$ being the distribution function of ${\vQ}$, which is continuous with strictly increasing marginal distribution functions on $[0,\infty)$ if (A3) holds.
For \Cref{LemmaMC}, we choose $F_n$ as empirical distribution function of the Monte-Carlo replicates ${\vQ}^{MC,1},\dots,{\vQ}^{MC,B_N}$.
For the bootstrap methods, we choose $F_N$ as empirical distribution function of the Monte-Carlo replicates ${\vQ}^{\star,1}_N,\dots,{\vQ}^{\star,B_N}_N$ resp. ${\vQ}^{\dagger,1}_N,\dots,{\vQ}^{\dagger,B_N}_N$.
Moreover, $F_{N,\ell},\ell=1,\dots,L,$ are chosen as the marginal distribution functions of $F_N$.
Then, the consistency of the empirical covariance matrix, \Cref{PBTheorem1} and \Cref{WBTheorem1} ensure the condition of Lemma~S7 in \cite{munko2024}, respectively, and, thus, (S10) therein follows.
Furthermore, Remark~1 in \cite{munko2024} yields (S11) and the function $\mathrm{FWER}$ in Lemma~S8 in the supplement of \cite{munko2024} can be shown to be strictly increasing.
Hence, Lemma~S8 in \cite{munko2024} is applicable and provides that the quantiles $q_{\ell,B_N}^{MC}, q_{\ell,B_N}^{\star}, q_{\ell,B_N}^{\dagger}$ all converge in probability to $q_\ell$ fulfilling $F(q_1,\dots,q_L) = 1 - \alpha$. Thus, the asymptotic correctness follows.
The asymptotic control of the family-wise error rate in the strong sense follows from
\begin{align*}
    \E\left[\max\limits_{\ell\in\mathcal T} \ind\left\{\left(\frac{Q_{N\ell}}{q_{\ell,B_N}^{MC}}\right)>1\right\}\right] &\to \E\left[\max\limits_{\ell\in\mathcal T} \ind\left\{\left(\frac{ Q_{\ell}}{q_{\ell}}\right)>1\right\}\right]
    \\& \leq \E\left[\max\limits_{\ell\in\{1,\dots,L\}} \ind\left\{\left(\frac{ Q_{\ell}}{q_{\ell}}\right)>1\right\}\right]
    \\& = 1 - F(q_1,\dots,q_L) = \alpha
\end{align*} for all index subsets $\mathcal T \subset\{1,\dots,L\}$ of true hypotheses. Analogously, the asymptotic control of the family-wise error rate in the strong sense holds for the bootstrap methods.

\section{Further simulations}

Although different quadratic forms can be used for QFMCT, up to now, only the ATS was used. Here, we additionally considered the QFMCT based on the WTS, as comparison and also consider the wild bootstrap. 
For a better overview, we remind you again of the tests considered:
\begin{itemize}
    \item $\varphi$: a classical multiple contrast test with a Tukey-type contrast matrix containing 24 local hypotheses, using equicoordinate quantile as well as parametric bootstrap ($\star$) or wild bootstrap ($\dagger$).
    \item $\varphi_{\vQ}$: the new quadratic form based multiple contrast test using WTS and ATS as quadratic forms and crictial values from a Monte-Carlo approach (MC), parametric bootstrap ($\star$) or wild bootstrap ($\dagger$).
    \item $\psi_{ATS}$: a single ATS for testing the global hypotheses  respectively one for each of the local hypothesis for the local test, which in this case are adjusted with the approach developed in \cite{holm1979}. For each of them a parametric bootstrap ($\star$)  or wild bootstrap ($\dagger$) is used.
\end{itemize}
\subsection{Global}
The results can be seen in \Cref{SupCSA1}-\ref{SupGSA2}. First of all, we want to focus on comparing both bootstrap approaches, parametric and wild. While for the classical ATS the wild bootstrap seems to be slightly better, for all multiple contrast test, the parametric bootstrap performs better. The wild bootstrap is mostly more liberal, while the difference gets smaller for increasing sample sizes. Of course, this often leads to even higher power compared with the other test, which allows the choice of the procedure in accordance with the respective focus. Also, the QFMCT based on the WTS performs well for both resampling techniques. In a direct comparison,
the ATS is slightly better overall regarding the type-I error rate as well as regarding the power.  
Moreover, the ATS has a shorter calculation time and less restrictive conditions (for the ATS condition (A3) is equivalent to $\vC_\ell\vSigma\vC_\ell^\top\neq \vnull_{r_\ell\times r_\ell}$). So, the usage of the WTS instead of the ATS for QFMCT is only reasonable if its other properties are required, like being invariant under scale transformations of the data.

   \begin{table}[htbp]
\begin{center}
\begin{scriptsize}
\begin{tabular}{ |l|r|r|r|r|r|r|r|r|r|r|} \hline
\rowcolor{ashgrey}  & {\hspace{-0.1cm} N\hspace{0.01cm}} & {\hspace{-0.1cm} $\delta=0$\hspace{-0.1cm}}&{\hspace{-0.1cm}  $\delta=0.25$\hspace{-0.1cm}}&{\hspace{-0.1cm}  $\delta=0.5$\hspace{-0.1cm}}&{ \hspace{-0.1cm}$\delta=0.75$\hspace{-0.1cm}}&{\hspace{-0.1cm}  $\delta=1.00$\hspace{-0.1cm}}& {\hspace{-0.1cm} $\delta=1.25$\hspace{-0.1cm}}& {\hspace{-0.1cm} $\delta=1.5$\hspace{-0.1cm}}&{ \hspace{-0.1cm} $\delta=1.75$\hspace{-0.1cm}}& {\hspace{-0.1cm} $\delta=2.00$\hspace{-0.1cm}}\\
\hline  \rowcolor{lightyellow}{$\varphi$} & 30 & 30.38 & 31.37 & 35.15 & 40.75 & 47.90 & 56.84 & 66.03 & 74.12 & 81.89 \\ [0.03cm] 
 \rowcolor{gainsboro} {$\varphi^{\star}$} & 30 & 6.01 & 6.13 & 7.86 & 9.68 & 13.49 & 18.50 & 24.63 & 31.67 & 39.27 \\ [0.03cm] 
  \rowcolor{lightyellow}{$\varphi^{\dagger}$} & 30 & 6.17 & 6.35 & 7.75 & 9.44 & 13.15 & 18.01 & 23.71 & 30.37 & 37.38 \\ [0.03cm] 
 \rowcolor{gainsboro} {$\varphi_{\vQ}^{MC}(ATS)$} & 30 & 15.89 & 16.90 & 19.79 & 25.17 & 32.42 & 41.48 & 50.68 & 61.58 & 71.19 \\ 
  \rowcolor{lightyellow}$\varphi_{\vQ}^\star(ATS)$ & 30 & 6.46 & 6.71 & 8.41 & 10.74 & 14.90 & 21.12 & 27.77 & 35.77 & 44.55 \\ 
 \rowcolor{gainsboro} $\varphi_{\vQ}^\dagger(ATS)$ & 30 & 8.54 & 8.73 & 10.73 & 14.10 & 18.74 & 25.68 & 33.15 & 42.02 & 51.19 \\ 
 \rowcolor{lightyellow} $\varphi_{\vQ}^\star(WTS)$ & 30 & 7.08 & 7.32 & 8.91 & 11.43 & 15.08 & 20.54 & 26.44 & 33.53 & 40.89 \\ 
 \rowcolor{gainsboro} $\varphi_{\vQ}^\dagger(WTS)$ & 30 & 10.31 & 10.52 & 12.37 & 15.15 & 20.00 & 26.01 & 33.18 & 40.57 & 48.46 \\ 
\rowcolor{lightyellow}  $\psi_{ATS}^\star$ & 30 & 2.66 & 2.87 & 4.36 & 6.71 & 11.54 & 17.73 & 26.17 & 36.92 & 49.16 \\ [0.03cm] 
 \rowcolor{gainsboro} $\psi_{ATS}^\dagger$ & 30 & \bf{4.95} & 5.47 & 8.01 & 11.78 & 17.80 & 26.09 & 36.00 & 46.97 & 59.50 \\ [0.05cm] \hline
\rowcolor{lightyellow}   {$\varphi$} & 60 & 14.78 & 17.16 & 23.94 & 33.74 & 47.46 & 63.68 & 77.58 & 88.02 & 94.46 \\ [0.03cm] 
\rowcolor{gainsboro}  {$\varphi^{\star}$} & 60 & 5.66 & 6.25 & 9.78 & 15.67 & 26.24 & 39.01 & 55.24 & 69.99 & 82.14 \\ [0.03cm] 
 \rowcolor{lightyellow} {$\varphi^{\dagger}$} & 60 & \bf{5.27} & 6.08 & 9.38 & 15.23 & 25.22 & 37.81 & 54.22 & 68.76 & 81.63 \\ [0.03cm] 
 \rowcolor{gainsboro} {$\varphi_{\vQ}^{MC}(ATS)$} & 60 & 9.85 & 11.19 & 16.93 & 25.90 & 40.61 & 55.80 & 71.69 & 84.33 & 92.68 \\ 
\rowcolor{lightyellow}  $\varphi_{\vQ}^\star(ATS)$ & 60 & 5.55 & 6.31 & 10.21 & 16.36 & 28.25 & 42.22 & 58.41 & 73.74 & 85.80 \\ 
 \rowcolor{gainsboro} $\varphi_{\vQ}^\dagger(ATS)$ & 60 & 6.02 & 6.88 & 11.03 & 17.54 & 30.01 & 43.44 & 60.09 & 74.69 & 86.35 \\ 
\rowcolor{lightyellow}  $\varphi_{\vQ}^\star(WTS)$ & 60 & 5.91 & 6.38 & 10.43 & 15.80 & 26.57 & 38.66 & 53.77 & 68.13 & 80.99 \\ 
 \rowcolor{gainsboro} $\varphi_{\vQ}^\dagger(WTS)$ & 60 & 5.90 & 6.62 & 10.15 & 16.05 & 26.01 & 38.43 & 53.30 & 67.70 & 80.64 \\ 
\rowcolor{lightyellow}  $\psi_{ATS}^\star$ & 60 & 3.85 & 4.81 & 9.54 & 17.76 & 33.13 & 50.33 & 68.81 & 83.92 & 93.02 \\ [0.05cm] 
 \rowcolor{gainsboro} $\psi_{ATS}^\dagger$ & 60 & 5.51 & 6.67 & 12.65 & 22.13 & 38.38 & 56.13 & 73.52 & 86.86 & 94.67 \\ [0.05cm] \hline
\rowcolor{lightyellow}   {$\varphi$} & 120 & 9.43 & 12.75 & 24.26 & 43.86 & 67.74 & 86.05 & 95.87 & 99.31 & 99.87 \\ [0.03cm] 
 \rowcolor{gainsboro} {$\varphi^{\star}$} & 120 & \bf{5.41} & 7.32 & 15.99 & 32.24 & 55.76 & 77.69 & 91.78 & 98.05 & 99.64 \\ [0.03cm] 
\rowcolor{lightyellow}  {$\varphi^{\dagger}$} & 120 & \bf{5.37} & 7.22 & 15.74 & 32.06 & 55.11 & 77.38 & 91.44 & 97.94 & 99.63 \\ [0.03cm] 
 \rowcolor{gainsboro} {$\varphi_{\vQ}^{MC}(ATS)$} & 120 & 6.95 & 10.39 & 20.73 & 40.35 & 64.67 & 83.84 & 94.90 & 99.00 & 99.87 \\ 
\rowcolor{lightyellow}  $\varphi_{\vQ}^\star(ATS)$ & 120 & \bf{4.95} & 7.90 & 16.32 & 34.06 & 57.67 & 79.04 & 92.55 & 98.50 & 99.73 \\ 
 \rowcolor{gainsboro} $\varphi_{\vQ}^\dagger(ATS)$ & 120 & \bf{4.87} & 8.08 & 16.87 & 34.43 & 58.24 & 79.08 & 92.82 & 98.64 & 99.71 \\ 
\rowcolor{lightyellow}  $\varphi_{\vQ}^\star(WTS)$ & 120 & 5.57 & 7.39 & 15.72 & 31.02 & 52.81 & 75.26 & 89.88 & 97.37 & 99.47 \\ 
 \rowcolor{gainsboro} $\varphi_{\vQ}^\dagger(WTS)$ & 120 & \bf{5.35} & 7.27 & 15.43 & 30.88 & 52.73 & 74.89 & 89.72 & 97.35 & 99.47 \\ 
\rowcolor{lightyellow}  $\psi_{ATS}^\star$ & 120 & 4.35 & 7.77 & 19.28 & 40.93 & 69.42 & 87.94 & 97.37 & 99.57 & 99.95 \\ [0.05cm] 
\rowcolor{gainsboro}  $\psi_{ATS}^\dagger$ & 120 & \bf{5.29} & 9.01 & 21.29 & 44.09 & 71.84 & 89.49 & 97.77 & 99.65 & 99.95 \\ [0.05cm] 
   \hline
\end{tabular}
\end{scriptsize}
\end{center}
\caption{Power of different tests for $\Hypo^A$ under a shift-alternative for 4 unbalanced groups
with  $n_1=n_2= N/3$ and $n_3=n_4= N/6$.
The 4-dimensional error terms are based on the standard normal distribution and covariance matrices $(\vSigma_1)_{j_1 j_2}=(\vSigma_2)_{j_1 j_2}=0.65^{|j_1-j_2|}$ respective  $\vSigma_3=\vSigma_4=\diag(3,4,5,6)+\veins_4\veins_4^\top$.}
\label{SupCSA1}
\end{table}

   \begin{table}[htbp]
\begin{center}
\begin{scriptsize}
\begin{tabular}{ |l|r|r|r|r|r|r|r|r|r|r|} \hline
\rowcolor{ashgrey}  & {\hspace{-0.1cm} N\hspace{0.01cm}} & {\hspace{-0.1cm} $\delta=0$\hspace{-0.1cm}}&{\hspace{-0.1cm}  $\delta=0.25$\hspace{-0.1cm}}&{\hspace{-0.1cm}  $\delta=0.5$\hspace{-0.1cm}}&{ \hspace{-0.1cm}$\delta=0.75$\hspace{-0.1cm}}&{\hspace{-0.1cm}  $\delta=1.00$\hspace{-0.1cm}}& {\hspace{-0.1cm} $\delta=1.25$\hspace{-0.1cm}}& {\hspace{-0.1cm} $\delta=1.5$\hspace{-0.1cm}}&{ \hspace{-0.1cm} $\delta=1.75$\hspace{-0.1cm}}& {\hspace{-0.1cm} $\delta=2.00$\hspace{-0.1cm}}\\
\hline \rowcolor{lightyellow} {$\varphi$} & 30 & 28.14 & 29.26 & 33.44 & 39.88 & 49.15 & 58.19 & 67.34 & 76.52 & 83.91 \\ [0.03cm] 
 \rowcolor{gainsboro} {$\varphi^{\star}$} & 30 & \bf{4.83} & 5.25 & 6.76 & 9.71 & 13.39 & 18.94 & 24.80 & 33.19 & 42.53 \\ [0.03cm] 
 \rowcolor{lightyellow} {$\varphi^{\dagger}$} & 30 & 5.55 & 5.95 & 7.25 & 10.61 & 14.22 & 19.80 & 26.16 & 33.48 & 42.52 \\ [0.03cm] 
 \rowcolor{gainsboro} {$\varphi_{\vQ}^{MC}(ATS)$} & 30 & 14.53 & 15.64 & 19.02 & 24.64 & 33.37 & 42.77 & 52.40 & 63.78 & 73.38 \\ 
\rowcolor{lightyellow}  $\varphi_{\vQ}^\star(ATS)$ & 30 & 5.73 & 5.94 & 7.59 & 10.70 & 15.62 & 21.46 & 28.33 & 38.29 & 48.17 \\ 
 \rowcolor{gainsboro} $\varphi_{\vQ}^\dagger(ATS)$ & 30 & 7.63 & 8.25 & 9.82 & 13.97 & 20.17 & 26.59 & 34.41 & 45.31 & 55.18 \\ 
\rowcolor{lightyellow}  $\varphi_{\vQ}^\star(WTS)$ & 30 & 5.75 & 6.17 & 7.96 & 11.10 & 15.51 & 21.29 & 27.11 & 35.43 & 44.45 \\ 
 \rowcolor{gainsboro} $\varphi_{\vQ}^\dagger(WTS)$ & 30 & 9.19 & 9.80 & 11.59 & 15.60 & 21.00 & 27.24 & 34.91 & 43.97 & 53.27 \\ 
\rowcolor{lightyellow}  $\psi_{ATS}^\star$ & 30 & 2.07 & 2.57 & 4.05 & 7.00 & 11.98 & 18.00 & 26.62 & 38.22 & 50.42 \\ [0.03cm] 
 \rowcolor{gainsboro} $\psi_{ATS}^\dagger$ & 30 & 4.54 & 5.40 & 7.55 & 11.55 & 17.94 & 26.82 & 36.49 & 48.45 & 60.75 \\ [0.05cm] \hline
 \rowcolor{lightyellow}  {$\varphi$} & 60 & 13.78 & 15.92 & 22.92 & 33.88 & 49.88 & 65.06 & 79.42 & 89.69 & 95.48 \\ [0.03cm] 
 \rowcolor{gainsboro} {$\varphi^{\star}$} & 60 & \bf{4.77} & 5.77 & 9.29 & 16.38 & 27.60 & 41.31 & 57.12 & 72.84 & 84.62 \\ [0.03cm] 
\rowcolor{lightyellow}  {$\varphi^{\dagger}$} & 60 & \bf{5.08} & 6.14 & 9.87 & 17.21 & 28.62 & 42.54 & 58.34 & 73.83 & 85.59 \\ [0.03cm] 
 \rowcolor{gainsboro} {$\varphi_{\vQ}^{MC}(ATS)$} & 60 & 8.58 & 10.72 & 16.48 & 26.78 & 41.10 & 56.50 & 72.43 & 85.23 & 93.29 \\ 
\rowcolor{lightyellow}  $\varphi_{\vQ}^\star(ATS)$ & 60 & \bf{4.61} & 6.04 & 9.71 & 16.79 & 28.51 & 42.31 & 59.01 & 74.87 & 86.02 \\ 
 \rowcolor{gainsboro} $\varphi_{\vQ}^\dagger(ATS)$ & 60 & \bf{5.30} & 6.85 & 11.03 & 18.57 & 31.25 & 45.79 & 61.72 & 76.91 & 87.79 \\ 
\rowcolor{lightyellow}  $\varphi_{\vQ}^\star(WTS)$ & 60 & \bf{5.02} & 6.07 & 9.60 & 16.68 & 27.71 & 41.20 & 56.21 & 71.58 & 83.37 \\ 
\rowcolor{gainsboro}  $\varphi_{\vQ}^\dagger(WTS)$ & 60 & \bf{5.35} & 6.37 & 10.36 & 17.56 & 28.27 & 42.51 & 57.01 & 72.59 & 84.42 \\ 
\rowcolor{lightyellow}  $\psi_{ATS}^\star$ & 60 & 3.30 & 4.66 & 8.84 & 18.13 & 32.43 & 49.82 & 67.79 & 82.88 & 92.18 \\ [0.03cm] 
 \rowcolor{gainsboro} $\psi_{ATS}^\dagger$ & 60 & \bf{5.08} & 6.58 & 12.05 & 23.23 & 38.48 & 56.31 & 73.79 & 86.96 & 94.28 \\ [0.05cm] \hline
 \rowcolor{lightyellow}  {$\varphi$} & 120 & 8.90 & 12.67 & 24.45 & 45.92 & 68.71 & 86.98 & 96.21 & 99.30 & 99.89 \\ [0.03cm] 
\rowcolor{gainsboro}  {$\varphi^{\star}$} & 120 & \bf{4.94} & 7.49 & 15.89 & 33.69 & 56.98 & 79.63 & 92.66 & 98.38 & 99.77 \\ [0.03cm] 
 \rowcolor{lightyellow} {$\varphi^{\dagger}$} & 120 & \bf{5.12} & 7.60 & 16.45 & 34.54 & 58.50 & 80.37 & 92.98 & 98.48 & 99.77 \\ [0.03cm] 
 \rowcolor{gainsboro} {$\varphi_{\vQ}^{MC}(ATS)$} & 120 & 7.25 & 10.01 & 21.17 & 41.49 & 64.19 & 84.56 & 95.34 & 98.99 & 99.83 \\ 
\rowcolor{lightyellow}  $\varphi_{\vQ}^\star(ATS)$ & 120 & \bf{5.04} & 7.41 & 16.43 & 34.58 & 57.05 & 80.10 & 92.95 & 98.30 & 99.70 \\ 
 \rowcolor{gainsboro} $\varphi_{\vQ}^\dagger(ATS)$ & 120 & \bf{5.42} & 8.09 & 17.14 & 36.28 & 58.50 & 81.42 & 93.66 & 98.48 & 99.73 \\ 
\rowcolor{lightyellow}  $\varphi_{\vQ}^\star(WTS)$ & 120 & \bf{5.00} & 7.51 & 15.48 & 32.68 & 54.66 & 76.93 & 91.11 & 97.72 & 99.60 \\ 
\rowcolor{gainsboro}  $\varphi_{\vQ}^\dagger(WTS)$ & 120 & \bf{5.18} & 7.53 & 16.15 & 33.43 & 55.24 & 77.81 & 91.42 & 97.78 & 99.58 \\ 
 \rowcolor{lightyellow} $\psi_{ATS}^\star$ & 120 & 4.19 & 7.42 & 18.87 & 41.14 & 66.78 & 87.73 & 96.88 & 99.49 & 99.98 \\ [0.03cm] 
\rowcolor{gainsboro}  $\psi_{ATS}^\dagger$ & 120 & \bf{5.20} & 8.75 & 21.49 & 44.94 & 70.11 & 89.41 & 97.45 & 99.61 & 99.98 \\ [0.05cm] 
   \hline
\end{tabular}
\end{scriptsize}
\end{center}
\caption{Power of different tests for $\Hypo^A$ under a shift-alternative for 4 unbalanced groups
with  $n_1=n_2= N/3$ and $n_3=n_4= N/6$.
The 4-dimensional error terms are based on the $t_9$ distribution and covariance matrices $(\vSigma_1)_{j_1 j_2}=(\vSigma_2)_{j_1 j_2}=0.65^{|j_1-j_2|}$ respective  $\vSigma_3=\vSigma_4=\diag(3,4,5,6)+\veins_4\veins_4^\top$.}
\label{SupCSA2}
\end{table}

  \begin{table}[htbp]
\begin{center}
\begin{scriptsize}
\begin{tabular}{ |l|r|r|r|r|r|r|r|r|r|r|} \hline
\rowcolor{ashgrey}  & {\hspace{-0.1cm} N\hspace{0.01cm}} & {\hspace{-0.1cm} $\delta=0$\hspace{-0.1cm}}&{\hspace{-0.1cm}  $\delta=0.25$\hspace{-0.1cm}}&{\hspace{-0.1cm}  $\delta=0.5$\hspace{-0.1cm}}&{ \hspace{-0.1cm}$\delta=0.75$\hspace{-0.1cm}}&{\hspace{-0.1cm}  $\delta=1.00$\hspace{-0.1cm}}& {\hspace{-0.1cm} $\delta=1.25$\hspace{-0.1cm}}& {\hspace{-0.1cm} $\delta=1.5$\hspace{-0.1cm}}&{ \hspace{-0.1cm} $\delta=1.75$\hspace{-0.1cm}}& {\hspace{-0.1cm} $\delta=2.00$\hspace{-0.1cm}}\\
\hline  \rowcolor{lightyellow}{$\varphi$} & 30 & 30.38 & 31.37 & 35.15 & 40.75 & 47.90 & 56.84 & 66.03 & 74.12 & 81.89 \\ [0.03cm] 
 \rowcolor{gainsboro} {$\varphi^{\star}$} & 30 & 6.01 & 6.13 & 7.86 & 9.68 & 13.49 & 18.50 & 24.63 & 31.67 & 39.27 \\ [0.03cm] 
 \rowcolor{lightyellow} {$\varphi^{\dagger}$} & 30 & 6.17 & 6.35 & 7.75 & 9.44 & 13.15 & 18.01 & 23.71 & 30.37 & 37.38 \\ [0.03cm] 
 \rowcolor{gainsboro} {$\varphi_{\vQ}^{MC}(ATS)$} & 30 & 15.89 & 16.90 & 19.79 & 25.17 & 32.42 & 41.48 & 50.68 & 61.58 & 71.19 \\ 
\rowcolor{lightyellow}  $\varphi_{\vQ}^\star(ATS)$ & 30 & 6.46 & 6.71 & 8.41 & 10.74 & 14.90 & 21.12 & 27.77 & 35.77 & 44.55 \\ 
\rowcolor{gainsboro}  $\varphi_{\vQ}^\dagger(ATS)$ & 30 & 8.54 & 8.73 & 10.73 & 14.10 & 18.74 & 25.68 & 33.15 & 42.02 & 51.19 \\ 
\rowcolor{lightyellow}  $\varphi_{\vQ}^\star(WTS)$ & 30 & 7.08 & 7.32 & 8.91 & 11.43 & 15.08 & 20.54 & 26.44 & 33.53 & 40.89 \\ 
 \rowcolor{gainsboro} $\varphi_{\vQ}^\dagger(WTS)$ & 30 & 10.31 & 10.52 & 12.37 & 15.15 & 20.00 & 26.01 & 33.18 & 40.57 & 48.46 \\ 
 \rowcolor{lightyellow} $\psi_{ATS}^\star$ & 30 & 2.66 & 2.87 & 4.36 & 6.71 & 11.54 & 17.73 & 26.17 & 36.92 & 49.16 \\ [0.03cm] 
 \rowcolor{gainsboro} $\psi_{ATS}^\dagger$ & 30 & \bf{4.95} & 5.47 & 8.01 & 11.78 & 17.80 & 26.09 & 36.00 & 46.97 & 59.50 \\ [0.05cm]  \hline
 \rowcolor{lightyellow}  {$\varphi$} & 60 & 14.78 & 17.16 & 23.94 & 33.74 & 47.46 & 63.68 & 77.58 & 88.02 & 94.46 \\ [0.03cm] 
 \rowcolor{gainsboro} {$\varphi^{\star}$} & 60 & 5.66 & 6.25 & 9.78 & 15.67 & 26.24 & 39.01 & 55.24 & 69.99 & 82.14 \\ [0.03cm] 
 \rowcolor{lightyellow} {$\varphi^{\dagger}$} & 60 & \bf{5.27} & 6.08 & 9.38 & 15.23 & 25.22 & 37.81 & 54.22 & 68.76 & 81.63 \\ [0.03cm] 
 \rowcolor{gainsboro} {$\varphi_{\vQ}^{MC}(ATS)$} & 60 & 9.85 & 11.19 & 16.93 & 25.90 & 40.61 & 55.80 & 71.69 & 84.33 & 92.68 \\ 
\rowcolor{lightyellow}  $\varphi_{\vQ}^\star(ATS)$ & 60 & 5.55 & 6.31 & 10.21 & 16.36 & 28.25 & 42.22 & 58.41 & 73.74 & 85.80 \\ 
 \rowcolor{gainsboro} $\varphi_{\vQ}^\dagger(ATS)$ & 60 & 6.02 & 6.88 & 11.03 & 17.54 & 30.01 & 43.44 & 60.09 & 74.69 & 86.35 \\ 
\rowcolor{lightyellow}  $\varphi_{\vQ}^\star(WTS)$ & 60 & 5.91 & 6.38 & 10.43 & 15.80 & 26.57 & 38.66 & 53.77 & 68.13 & 80.99 \\ 
 \rowcolor{gainsboro} $\varphi_{\vQ}^\dagger(WTS)$ & 60 & 5.90 & 6.62 & 10.15 & 16.05 & 26.01 & 38.43 & 53.30 & 67.70 & 80.64 \\ 
\rowcolor{lightyellow}  $\psi_{ATS}^\star$ & 60 & 3.85 & 4.81 & 9.54 & 17.76 & 33.13 & 50.33 & 68.81 & 83.92 & 93.02 \\ [0.03cm]
\rowcolor{gainsboro}  $\psi_{ATS}^\dagger$ & 60 & 5.51 & 6.67 & 12.65 & 22.13 & 38.38 & 56.13 & 73.52 & 86.86 & 94.67 \\ [0.05cm]\hline 
\rowcolor{lightyellow}  {$\varphi$} & 120 & 9.43 & 12.75 & 24.26 & 43.86 & 67.74 & 86.05 & 95.87 & 99.31 & 99.87 \\ [0.03cm] 
 \rowcolor{gainsboro} {$\varphi^{\star}$} & 120 & \bf{5.41} & 7.32 & 15.99 & 32.24 & 55.76 & 77.69 & 91.78 & 98.05 & 99.64 \\ [0.03cm] 
\rowcolor{lightyellow}  {$\varphi^{\dagger}$} & 120 & \bf{5.37} & 7.22 & 15.74 & 32.06 & 55.11 & 77.38 & 91.44 & 97.94 & 99.63 \\ [0.03cm] 
 \rowcolor{gainsboro}  {$\varphi_{\vQ}^{MC}(ATS)$}  & 120 &6.95 & 10.39 & 20.73 & 40.35 & 64.67 & 83.84 & 94.90 & 99.00 & 99.87 \\ 
\rowcolor{lightyellow} $\varphi_{\vQ}^\star(ATS)$ & 120 & \bf{4.95} & 7.90 & 16.32 & 34.06 & 57.67 & 79.04 & 92.55 & 98.50 & 99.73 \\ 
 \rowcolor{gainsboro} $\varphi_{\vQ}^\dagger(ATS)$ & 120 & \bf{4.87} & 8.08 & 16.87 & 34.43 & 58.24 & 79.08 & 92.82 & 98.64 & 99.71 \\ 
\rowcolor{lightyellow}  $\varphi_{\vQ}^\star(WTS)$ & 120 & 5.57 & 7.39 & 15.72 & 31.02 & 52.81 & 75.26 & 89.88 & 97.37 & 99.47 \\ 
 \rowcolor{gainsboro} $\varphi_{\vQ}^\dagger(WTS)$ & 120 & \bf{5.35} & 7.27 & 15.43 & 30.88 & 52.73 & 74.89 & 89.72 & 97.35 & 99.47 \\ 
\rowcolor{lightyellow}  $\psi_{ATS}^\star$ & 120 & 4.35 & 7.77 & 19.28 & 40.93 & 69.42 & 87.94 & 97.37 & 99.57 & 99.95 \\ [0.03cm] 
 \rowcolor{gainsboro} $\psi_{ATS}^\dagger$ & 120 & \bf{5.29} & 9.01 & 21.29 & 44.09 & 71.84 & 89.49 & 97.77 & 99.65 & 99.95 \\ [0.05cm] 
   \hline
\end{tabular}
\end{scriptsize}
\end{center}
\caption{Power of different tests for $\Hypo^B$ under a shift-alternative for 4 unbalanced groups
with  $n_1=n_2= N/3$ and $n_3=n_4= N/6$.
The 4-dimensional error terms are based on the standard normal distribution and covariance matrices $(\vSigma_1)_{j_1 j_2}=(\vSigma_2)_{j_1 j_2}=0.65^{|j_1-j_2|}$ respective  $\vSigma_3=\vSigma_4=\diag(3,4,5,6)+\veins_4\veins_4^\top$.}
\label{SupGSA1}
\end{table}

   \begin{table}[htbp]
\begin{center}
\begin{scriptsize}
\begin{tabular}{ |l|r|r|r|r|r|r|r|r|r|r|} \hline
\rowcolor{ashgrey}  & {\hspace{-0.1cm} N\hspace{0.01cm}} & {\hspace{-0.1cm} $\delta=0$\hspace{-0.1cm}}&{\hspace{-0.1cm}  $\delta=0.25$\hspace{-0.1cm}}&{\hspace{-0.1cm}  $\delta=0.5$\hspace{-0.1cm}}&{ \hspace{-0.1cm}$\delta=0.75$\hspace{-0.1cm}}&{\hspace{-0.1cm}  $\delta=1.00$\hspace{-0.1cm}}& {\hspace{-0.1cm} $\delta=1.25$\hspace{-0.1cm}}& {\hspace{-0.1cm} $\delta=1.5$\hspace{-0.1cm}}&{ \hspace{-0.1cm} $\delta=1.75$\hspace{-0.1cm}}& {\hspace{-0.1cm} $\delta=2.00$\hspace{-0.1cm}}\\
\hline \rowcolor{lightyellow} {$\varphi$} & 30 & 28.14 & 29.26 & 33.44 & 39.88 & 49.15 & 58.19 & 67.34 & 76.52 & 83.91 \\ [0.03cm] 
 \rowcolor{gainsboro} {$\varphi^{\star}$} & 30 & \bf{4.83} & 5.25 & 6.76 & 9.71 & 13.39 & 18.94 & 24.80 & 33.19 & 42.53 \\ [0.03cm] 
\rowcolor{lightyellow}  {$\varphi^{\dagger}$} & 30 & 5.55 & 5.95 & 7.25 & 10.61 & 14.22 & 19.80 & 26.16 & 33.48 & 42.52 \\ [0.03cm] 
 \rowcolor{gainsboro} {$\varphi_{\vQ}^{MC}(ATS)$} & 30 & 14.53 & 15.64 & 19.02 & 24.64 & 33.37 & 42.77 & 52.40 & 63.78 & 73.38 \\ 
\rowcolor{lightyellow}  $\varphi_{\vQ}^\star(ATS)$ & 30 & 5.73 & 5.94 & 7.59 & 10.70 & 15.62 & 21.46 & 28.33 & 38.29 & 48.17 \\ 
 \rowcolor{gainsboro} $\varphi_{\vQ}^\dagger(ATS)$ & 30 & 7.63 & 8.25 & 9.82 & 13.97 & 20.17 & 26.59 & 34.41 & 45.31 & 55.18 \\ 
 \rowcolor{lightyellow} $\varphi_{\vQ}^\star(WTS)$ & 30 & 5.75 & 6.17 & 7.96 & 11.10 & 15.51 & 21.29 & 27.11 & 35.43 & 44.45 \\ 
 \rowcolor{gainsboro} $\varphi_{\vQ}^\dagger(WTS)$ & 30 & 9.19 & 9.80 & 11.59 & 15.60 & 21.00 & 27.24 & 34.91 & 43.97 & 53.27 \\ 
\rowcolor{lightyellow}  $\psi_{ATS}^\star$ & 30 & 2.07 & 2.57 & 4.05 & 7.00 & 11.98 & 18.00 & 26.62 & 38.22 & 50.42 \\ [0.03cm] 
 \rowcolor{gainsboro} $\psi_{ATS}^\dagger$ & 30 & 4.54 & 5.40 & 7.55 & 11.55 & 17.94 & 26.82 & 36.49 & 48.45 & 60.75 \\ [0.05cm] \hline
 \rowcolor{lightyellow}{$\varphi$} & 60 & 13.78 & 15.92 & 22.92 & 33.88 & 49.88 & 65.06 & 79.42 & 89.69 & 95.48 \\ [0.03cm] 
\rowcolor{gainsboro}  {$\varphi^{\star}$} & 60 & \bf{4.77} & 5.77 & 9.29 & 16.38 & 27.60 & 41.31 & 57.12 & 72.84 & 84.62 \\ [0.03cm] 
\rowcolor{lightyellow}  {$\varphi^{\dagger}$} & 60 & \bf{5.08} & 6.14 & 9.87 & 17.21 & 28.62 & 42.54 & 58.34 & 73.83 & 85.59 \\ [0.03cm] 
\rowcolor{gainsboro}  {$\varphi_{\vQ}^{MC}(ATS)$} & 60 & 8.58 & 10.72 & 16.48 & 26.78 & 41.10 & 56.50 & 72.43 & 85.23 & 93.29 \\ 
\rowcolor{lightyellow}  $\varphi_{\vQ}^\star(ATS)$ & 60 & \bf{4.61} & 6.04 & 9.71 & 16.79 & 28.51 & 42.31 & 59.01 & 74.87 & 86.02 \\ 
\rowcolor{gainsboro}  $\varphi_{\vQ}^\dagger(ATS)$ & 60 & \bf{5.30} & 6.85 & 11.03 & 18.57 & 31.25 & 45.79 & 61.72 & 76.91 & 87.79 \\ 
 \rowcolor{lightyellow} $\varphi_{\vQ}^\star(WTS)$ & 60 & \bf{5.02} & 6.07 & 9.60 & 16.68 & 27.71 & 41.20 & 56.21 & 71.58 & 83.37 \\ 
 \rowcolor{gainsboro} $\varphi_{\vQ}^\dagger(WTS)$ & 60 & \bf{5.35} & 6.37 & 10.36 & 17.56 & 28.27 & 42.51 & 57.01 & 72.59 & 84.42 \\ 
\rowcolor{lightyellow}  $\psi_{ATS}^\star$ & 60 & 3.30 & \bf{4.66} & 8.84 & 18.13 & 32.43 & 49.82 & 67.79 & 82.88 & 92.18 \\ [0.03cm] 
\rowcolor{gainsboro}  $\psi_{ATS}^\dagger$ & 60 & \bf{5.08} & 6.58 & 12.05 & 23.23 & 38.48 & 56.31 & 73.79 & 86.96 & 94.28 \\ [0.05cm] \hline
  \rowcolor{lightyellow}{$\varphi$} & 120 & 8.90 & 12.67 & 24.45 & 45.92 & 68.71 & 86.98 & 96.21 & 99.30 & 99.89 \\ [0.03cm] 
 \rowcolor{gainsboro} {$\varphi^{\star}$} & 120 & \bf{4.94} & 7.49 & 15.89 & 33.69 & 56.98 & 79.63 & 92.66 & 98.38 & 99.77 \\ [0.03cm] 
\rowcolor{lightyellow}  {$\varphi^{\dagger}$} & 120 & \bf{5.12} & 7.60 & 16.45 & 34.54 & 58.50 & 80.37 & 92.98 & 98.48 & 99.77 \\ [0.03cm] 
 \rowcolor{gainsboro} {$\varphi_{\vQ}^{MC}(ATS)$} & 120 & 7.25 & 10.01 & 21.17 & 41.49 & 64.19 & 84.56 & 95.34 & 98.99 & 99.83 \\ 
\rowcolor{lightyellow}  $\varphi_{\vQ}^\star(ATS)$ & 120 & \bf{5.04} & 7.41 & 16.43 & 34.58 & 57.05 & 80.10 & 92.95 & 98.30 & 99.70 \\ 
 \rowcolor{gainsboro} $\varphi_{\vQ}^\dagger(ATS)$ & 120 & \bf{5.42} & 8.09 & 17.14 & 36.28 & 58.50 & 81.42 & 93.66 & 98.48 & 99.73 \\ 
\rowcolor{lightyellow}  $\varphi_{\vQ}^\star(WTS)$ & 120 & \bf{5.00} & 7.51 & 15.48 & 32.68 & 54.66 & 76.93 & 91.11 & 97.72 & 99.60 \\ 
 \rowcolor{gainsboro} $\varphi_{\vQ}^\dagger(WTS)$ & 120 & \bf{5.18} & 7.53 & 16.15 & 33.43 & 55.24 & 77.81 & 91.42 & 97.78 & 99.58 \\ 
\rowcolor{lightyellow}  $\psi_{ATS}^\star$ & 120 & 4.19 & 7.42 & 18.87 & 41.14 & 66.78 & 87.73 & 96.88 & 99.49 & 99.98 \\ [0.03cm] 
  \rowcolor{gainsboro}$\psi_{ATS}^\dagger$ & 120 & \bf{5.20} & 8.75 & 21.49 & 44.94 & 70.11 & 89.41 & 97.45 & 99.61 & 99.98 \\ [0.05cm] 
   \hline
\end{tabular}
\end{scriptsize}
\end{center}
\caption{Power of different tests for $\Hypo^B$ under a shift-alternative for 4 unbalanced groups
with  $n_1=n_2= N/3$ and $n_3=n_4= N/6$.
The 4-dimensional error terms are based on the $t_9$ distribution and covariance matrices $(\vSigma_1)_{j_1 j_2}=(\vSigma_2)_{j_1 j_2}=0.65^{|j_1-j_2|}$ respective  $\vSigma_3=\vSigma_4=\diag(3,4,5,6)+\veins_4\veins_4^\top$.}
\label{SupGSA2}
\end{table}

\subsection{Local}
In \Cref{tab:CSAl1} and \Cref{tab:CSAl2} it can be seen, that the QFMCT with parametric bootstrap regonizes the violated local alternatives a bit worse than the ANOVA-type statistic with the Bonferroni-Holm adjustment. Due to the liberal behaviour, although for the Monte-Carlo approach all values are higher, the QFMCT is only recommendable here for large sample sizes.

  \begin{table}[htbp]
\begin{center}
\begin{scriptsize}
\begin{tabular}{ |l|r|r|r|r|r|r|r|r|r|} \hline
\rowcolor{ashgrey}  & {\hspace{-0.1cm} N\hspace{0.01cm}} &{\hspace{-0.1cm}  $\delta=0.25$\hspace{-0.1cm}}&{\hspace{-0.1cm}  $\delta=0.5$\hspace{-0.1cm}}&{ \hspace{-0.1cm}$\delta=0.75$\hspace{-0.1cm}}&{\hspace{-0.1cm}  $\delta=1.00$\hspace{-0.1cm}}& {\hspace{-0.1cm} $\delta=1.25$\hspace{-0.1cm}}& {\hspace{-0.1cm} $\delta=1.5$\hspace{-0.1cm}}&{ \hspace{-0.1cm} $\delta=1.75$\hspace{-0.1cm}}& {\hspace{-0.1cm} $\delta=2.00$\hspace{-0.1cm}}\\\hline 

 \rowcolor{gainsboro}  {$\varphi_{\vQ}^{MC}(ATS)$}-Any & 30 & 16.90 & 19.79 & 25.17 & 32.42 & 41.48 & 50.68 & 61.58 & 71.19 \\ 
 \rowcolor{lightyellow}  {$\varphi_{\vQ}^\star(ATS)$}-Any & 30 & 6.71 & 8.41 & 10.74 & 14.90 & 21.12 & 27.77 & 35.77 & 44.55 \\ 
 \rowcolor{gainsboro} $\psi_{ATS}^\star(BH)$-Any & 30 & 9.68 & 11.87 & 15.09 & 20.17 & 27.36 & 34.56 & 43.29 & 52.95 \\ 

\rowcolor{lightyellow}   {$\varphi_{\vQ}^{MC}(ATS)$}-Sim & 30 & 0.00 & 0.02 & 0.01 & 0.13 & 0.16 & 0.40 & 0.94 & 2.05 \\ 
  \rowcolor{gainsboro} {$\varphi_{\vQ}^\star(ATS)$}-Sim & 30 & 0.00 & 0.00 & 0.00 & 0.00 & 0.03 & 0.05 & 0.09 & 0.26 \\
  \rowcolor{lightyellow}  $\psi_{ATS}^\star(BH)$-Sim & 30 & 0.00 & 0.02 & 0.01 & 0.04 & 0.12 & 0.25 & 0.50 & 1.18 \\ \hline

  \rowcolor{gainsboro} {$\varphi_{\vQ}^{MC}(ATS)$}-Any & 60 & 11.19 & 16.93 & 25.90 & 40.61 & 55.80 & 71.69 & 84.33 & 92.68 \\ 
\rowcolor{lightyellow}   {$\varphi_{\vQ}^\star(ATS)$}-Any & 60 & 6.31 & 10.21 & 16.36 & 28.25 & 42.22 & 58.41 & 73.74 & 85.80 \\ 
 \rowcolor{gainsboro}  $\psi_{ATS}^\star(BH)$-Any & 60 & 9.97 & 14.44 & 22.10 & 35.21 & 49.93 & 65.05 & 79.26 & 89.84 \\ 

\rowcolor{lightyellow}   {$\varphi_{\vQ}^{MC}(ATS)$}-Sim & 60 & 0.01 & 0.01 & 0.06 & 0.24 & 0.85 & 2.49 & 5.88 & 13.30 \\ 
  \rowcolor{gainsboro} {$\varphi_{\vQ}^\star(ATS)$}-Sim & 60 & 0.00 & 0.00 & 0.01 & 0.07 & 0.30 & 0.83 & 2.32 & 6.41 \\ 
   \rowcolor{lightyellow}  $\psi_{ATS}^\star(BH)$-Sim & 60 & 0.01 & 0.01 & 0.09 & 0.33 & 1.26 & 3.14 & 7.53 & 16.22 \\ \hline

  \rowcolor{gainsboro} {$\varphi_{\vQ}^{MC}(ATS)$}-Any & 120 & 10.39 & 20.73 & 40.35 & 64.67 & 83.84 & 94.90 & 99.00 & 99.87 \\ 
\rowcolor{lightyellow}   {$\varphi_{\vQ}^\star(ATS)$}-Any & 120 & 7.90 & 16.32 & 34.06 & 57.67 & 79.04 & 92.55 & 98.50 & 99.73 \\ 
 \rowcolor{gainsboro}   $\psi_{ATS}^\star(BH)$-Any & 120 & 11.59 & 22.22 & 41.35 & 64.30 & 83.15 & 94.75 & 99.13 & 99.86 \\ 

\rowcolor{lightyellow}   {$\varphi_{\vQ}^{MC}(ATS)$}-Sim & 120 & 0.00 & 0.01 & 0.20 & 1.62 & 6.32 & 17.29 & 36.38 & 58.31 \\ 
 \rowcolor{gainsboro}  {$\varphi_{\vQ}^\star(ATS)$}-Sim & 120 & 0.00 & 0.01 & 0.16 & 0.95 & 4.14 & 13.12 & 29.61 & 50.70 \\ 
  \rowcolor{lightyellow}  $\psi_{ATS}^\star(BH)$-Sim & 120  & 0.01 & 0.11 & 0.56 & 3.12 & 10.95 & 26.86 & 49.97 & 71.81 \\ 
  
 \hline
\end{tabular}
\end{scriptsize}
\end{center}
\caption{Simultaneous power and any power of local tests for $\Hypo^A$ under a shift-alternative for 4 unbalanced groups
with  $n_1=n_2= N/3$ and $n_3=n_4= N/6$.
The 4-dimensional error terms are based on the standard normal distribution and covariance matrices $(\vSigma_1)_{j_1 j_2}=(\vSigma_2)_{j_1 j_2}=0.65^{|j_1-j_2|}$ respective  $\vSigma_3=\vSigma_4=\diag(3,4,5,6)+\veins_4\veins_4^\top$.}
\label{tab:CSAl1}
\end{table}

  \begin{table}[htbp]
\begin{center}
\begin{scriptsize}
\begin{tabular}{ |l|r|r|r|r|r|r|r|r|r|r|} \hline
\rowcolor{ashgrey}  & {\hspace{-0.1cm} N\hspace{0.01cm}} & {\hspace{-0.1cm} $\delta=0$\hspace{-0.1cm}}&{\hspace{-0.1cm}  $\delta=0.25$\hspace{-0.1cm}}&{\hspace{-0.1cm}  $\delta=0.5$\hspace{-0.1cm}}&{ \hspace{-0.1cm}$\delta=0.75$\hspace{-0.1cm}}&{\hspace{-0.1cm}  $\delta=1.00$\hspace{-0.1cm}}& {\hspace{-0.1cm} $\delta=1.25$\hspace{-0.1cm}}& {\hspace{-0.1cm} $\delta=1.5$\hspace{-0.1cm}}&{ \hspace{-0.1cm} $\delta=1.75$\hspace{-0.1cm}}& {\hspace{-0.1cm} $\delta=2.00$\hspace{-0.1cm}}\\\hline 
\rowcolor{lightyellow}$\psi_{ATS}^\star(BH)$-Any & 30 & 8.44 & 8.74 & 10.84 & 14.90 & 20.88 & 27.86 & 35.53 & 45.86 & 56.41 \\ 
 \rowcolor{gainsboro}  {$\varphi_{\vQ}^{MC}(ATS)$}-Any & 30 & 14.53 & 15.64 & 19.02 & 24.64 & 33.37 & 42.77 & 52.40 & 63.78 & 73.38 \\ 
\rowcolor{lightyellow}   {$\varphi_{\vQ}^\star(ATS)$}-Any & 30 & 5.73 & 5.94 & 7.59 & 10.70 & 15.62 & 21.46 & 28.33 & 38.29 & 48.17 \\ 
  \rowcolor{gainsboro}$\psi_{ATS}^\star(BH)$-Sim & 30 & 0.00 & 0.00 & 0.01 & 0.01 & 0.03 & 0.11 & 0.28 & 0.53 & 1.01 \\ 
\rowcolor{lightyellow}   {$\varphi_{\vQ}^{MC}(ATS)$}-Sim & 30 & 0.00 & 0.01 & 0.01 & 0.02 & 0.10 & 0.28 & 0.55 & 1.17 & 2.08 \\ 
 \rowcolor{gainsboro}  {$\varphi_{\vQ}^\star(ATS)$}-Sim & 30 & 0.00 & 0.00 & 0.00 & 0.00 & 0.00 & 0.03 & 0.03 & 0.13 & 0.10 \\ \hline
\rowcolor{lightyellow}  $\psi_{ATS}^\star(BH)$-Any & 60 & 7.56 & 9.34 & 14.06 & 22.64 & 35.38 & 49.84 & 66.13 & 80.36 & 89.92 \\ 
 \rowcolor{gainsboro}  {$\varphi_{\vQ}^{MC}(ATS)$}-Any & 60 & 8.58 & 10.72 & 16.48 & 26.78 & 41.10 & 56.50 & 72.43 & 85.23 & 93.29 \\ 
\rowcolor{lightyellow}   {$\varphi_{\vQ}^\star(ATS)$}-Any & 60 & 4.61 & 6.04 & 9.71 & 16.79 & 28.51 & 42.31 & 59.01 & 74.87 & 86.02 \\ 
\rowcolor{gainsboro}  $\psi_{ATS}^\star(BH)$-Sim & 60 & 0.00 & 0.00 & 0.00 & 0.11 & 0.31 & 1.45 & 3.21 & 8.46 & 16.78 \\ 
\rowcolor{lightyellow}   {$\varphi_{\vQ}^{MC}(ATS)$}-Sim & 60 & 0.00 & 0.00 & 0.00 & 0.07 & 0.23 & 0.96 & 2.22 & 6.56 & 13.62 \\ 
 \rowcolor{gainsboro}  {$\varphi_{\vQ}^\star(ATS)$}-Sim & 60 & 0.00 & 0.00 & 0.00 & 0.01 & 0.01 & 0.32 & 0.73 & 2.76 & 6.39 \\ \hline
\rowcolor{lightyellow}  $\psi_{ATS}^\star(BH)$-Any & 120 & 7.89 & 11.06 & 21.87 & 41.63 & 63.57 & 84.16 & 95.19 & 98.97 & 99.79 \\ 
\rowcolor{gainsboro}   {$\varphi_{\vQ}^{MC}(ATS)$}-Any & 120 & 7.25 & 10.01 & 21.17 & 41.49 & 64.19 & 84.56 & 95.34 & 98.99 & 99.83 \\ 
\rowcolor{lightyellow}   {$\varphi_{\vQ}^\star(ATS)$}-Any & 120 & 5.04 & 7.41 & 16.43 & 34.58 & 57.05 & 80.10 & 92.95 & 98.30 & 99.70 \\ 
\rowcolor{gainsboro}  $\psi_{ATS}^\star(BH)$-Sim & 120 & 0.00 & 0.00 & 0.13 & 0.65 & 3.24 & 11.33 & 27.57 & 49.74 & 70.09 \\ 
\rowcolor{lightyellow}   {$\varphi_{\vQ}^{MC}(ATS)$}-Sim & 120 & 0.00 & 0.00 & 0.03 & 0.30 & 1.59 & 6.31 & 17.75 & 36.28 & 57.16 \\ 
\rowcolor{gainsboro}   {$\varphi_{\vQ}^\star(ATS)$}-Sim & 120 & 0.00 & 0.00 & 0.02 & 0.15 & 0.94 & 4.29 & 12.92 & 29.42 & 49.69 \\ [0.08cm]
 \hline
\end{tabular}
\end{scriptsize}
\end{center}
\caption{Simultaneous power and any power of local tests for $\Hypo^A$ under a shift-alternative for 4 unbalanced groups
with  $n_1=n_2= N/3$ and $n_3=n_4= N/6$.
The 4-dimensional error terms are based on the $t_9$ distribution and covariance matrices $(\vSigma_1)_{j_1 j_2}=(\vSigma_2)_{j_1 j_2}=0.65^{|j_1-j_2|}$ respective  $\vSigma_3=\vSigma_4=\diag(3,4,5,6)+\veins_4\veins_4^\top$.}
\label{tab:CSAl2}
\end{table}

\end{document}